\documentclass[11pt]{article}
\usepackage[a4paper,includeheadfoot,margin=2.54cm]{geometry}

\usepackage[latin1, utf8]{inputenc}

\usepackage{amsmath}
\usepackage{amssymb,amsthm,graphicx}
\usepackage[dvips]{epsfig}
\usepackage{amsfonts}
\usepackage{xcolor}
\usepackage{ulem}
\usepackage{hyperref}
\usepackage{enumerate}
\usepackage{caption}
\makeatletter
\newtheorem*{rep@theorem}{\rep@title}
\newcommand{\newreptheorem}[2]{%
\newenvironment{rep#1}[1]{%
 \def\rep@title{#2 \ref{##1}}%
 \begin{rep@theorem}}%
 {\end{rep@theorem}}}
\makeatother

\newtheorem{theorem}{Theorem}
\newtheorem{proposition}[theorem]{Proposition}
\newtheorem{corollary}[theorem]{Corollary}

\newtheorem{question}[theorem]{Question}
\newtheorem{lemma}[theorem]{Lemma}

\theoremstyle{remark}

\theoremstyle{definition}
\newtheorem{definition}[theorem]{Definition}

\newreptheorem{theorem}{Theorem}
\newreptheorem{proposition}{Proposition}

\newcommand{\N}{\mathbb{N}}

\definecolor{darkgreen}{rgb}{0,.5,0}

\newcounter{sideremark}

\newcommand{\new}[1]{{\color{red} #1}}

\title{How Expressive Are {Friendly School Partitions}?
    \thanks{This work is a part of a project
    (including COSP REU 2020) that has received funding
    from the European Union’s Horizon 2020 research and innovation programme
    under the Marie Skłodowska-Curie grant agreement No.~823748.
    M.S. was also supported by GA\v{C}R grant 22-19073S.
}}

\author{
Josef Mina\v{r}\'{i}k
\\ \small{Department of Applied Mathematics, Charles University,}
\and
Shay Moran
\\ \small{Department of Mathematics, Technion and Google Research,}
\and
Michael Skotnica
\\ \small{Department of Applied Mathematics, Charles University.}
\and
\\ This work is a result of a project started at the 2020 COSP REU summer school.
}

\date{}

\begin{document}
\maketitle

\begin{abstract}
A natural procedure for assigning students to classes in the beginning of the school-year 
    is to let each student write down a list of $d$ other students 
    with whom she/he wants to be in the same class (typically $d=3$). 
    The teachers then gather all the lists and try to assign 
    the students to classes in a way that each student is assigned to the same class 
    with at least one student from her/his list. 
    We refer to such partitions as {\it friendly}. 
    In realistic scenarios, the teachers may also consider other
    constraints when picking the friendly partition: 
    e.g.\ there may be a group of students whom the teachers wish to avoid 
    assigning to the same class;
    alternatively, there may be two close friends
    whom the teachers want to put together; etc.

 Inspired by such challenges, we explore questions
    concerning the {\it expressiveness} of friendly partitions. 
    For example: Does there always exist a friendly partition?
    More generally, how many friendly partitions are there?
    Can every student $u$ be separated from any other student~$v$?
    Does there exist a student $u$ that can be separated from any other student $v$?

We show that when $d\geq 3$ there always exist at least $2$ friendly partitions
    and when~$d\geq 15$ there always exists a student $u$ which can be separated
    from any other student $v$.
    The question regarding separability of each pair of students is left open,
    but we give a positive answer under the additional assumption that each student appears in at most roughly $\exp(d)$ lists.
    We further suggest several open questions and present some preliminary findings towards resolving them.
\end{abstract}

\section{Introduction}

In many schools the following procedure is used to assign students to classes:
    each student~$u$ writes down a list $L(u)$ of $d$ other students 
    with whom he or she wants to be in the same class. 
    The goal is to find a partition of the students to classes so that 
    each student is in the same class with at least $r\leq d$ students from $L(u)$.
    In realistic cases, $r=1$ and $d$ is a small constant, say $d=3$. 
    Of course, there may also be various other constraints the partition needs to satisfy 
    such as that the number of classes needs to be some parameter $k$ and 
    that all $k$ classes need to have roughly the same size,
    or constraints regarding certain groups of students
    which should/should not be put together, etc.

One can naturally model this problem in the language of graph-theory:
    define the {\it preferences graph} $G=(V,E)$ to be a directed graph
    where $V$ is the set of students and $u\to v \in E$ if and only if $v\in L(u)$. 
    So, the out-degree of every vertex in $G$ is $d$.
    The goal is then to find a partition of $V$
    such that the subdigraph induced by each part has minimum out-degree $\geq r$. 
    We will refer to such partitions as $r$-friendly partitions.
    Note that the trivial partition
    where all vertices are in the same part is $d$-friendly.

\subsection{Existence}
Perhaps the most basic question is whether non-trivial $r$-friendly partitions exist.
    Consider the case of $r=d=1$, and let $G=C_n$ be a directed cycle on $n$ vertices.
    See Figure~\ref{fig:C_6}.
\begin{figure}
    \centering
    \includegraphics[scale=.8]{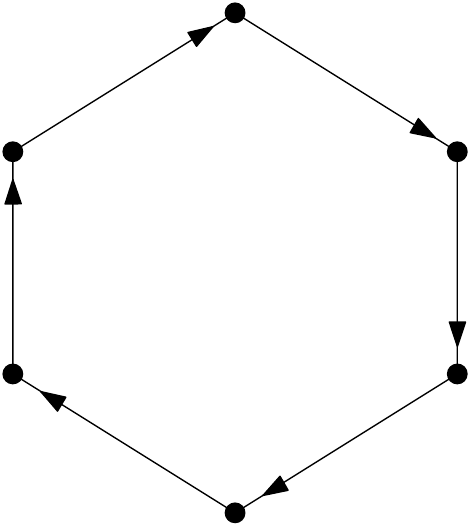}
    \caption{\small A directed cycle on 6 vertices. In any friendly partition,
        each vertex has to be in the same part with its unique out-neighbor.
        Hence, by transitivity, only the trivial partition is friendly.}
        \label{fig:C_6}
\end{figure}

Clearly, in this case, the only $r$-friendly partition is the trivial one.
    Thus, in order to guarantee the existence of non-trivial $1$-friendly partitions,
    $d$ has to be larger than $1$.
    Next, assume $d=2, r=1$.
    Also here there are digraphs for which only the trivial partition is $r$-friendly.
    A simple example for such a digraph is the complete directed graph on $3$ vertices.
    See Figure~\ref{fig:K_3}.
\begin{figure}
    \centering
    \includegraphics[scale=.8]{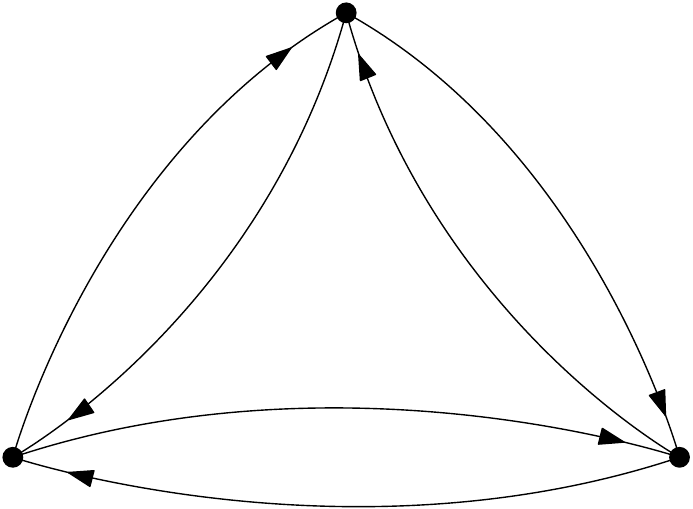}
    \caption{\small The complete directed graph on 3 vertices.
        Any non-trivial partition will leave one of the vertices alone in its part,
        and hence cannot be friendly.}
        \label{fig:K_3}
\end{figure}

Alon showed in~\cite{Alon20} that for $d\geq 3, r=1$
    there always exist a non-trivial friendly partition.
    The argument hinges on the following classical result due to Thomassen:

\begin{theorem}[\cite{Thomassen83}]\label{thm:thomassen}
    Each directed graph with minimum out-degree 3 contains two disjoint cycles.
\end{theorem}

Indeed, given two disjoint cycles $C_1, C_2$
    one can extend them into a $1$-friendly partition as follows.
    One part consists of the cycle $C_1$ and of all vertices
    from which there is a path to $C_1$ which does not intersect $C_2$,
    and the second part consists of all other vertices.
    It is easy to verify that this partition is indeed $1$-friendly.
    This argument leads to the following useful lemma 
    which we will use throughout the article:
\begin{lemma}\label{lem:extending_friendly_sets}
    Let $G=(V,E)$ be a digraph with minimal out-degree $\geq 1$,
    and let $U,W\subseteq V$ be disjoint sets which are \emph{friendly} in the following sense: each vertex in $U$ (resp.\ $W$) has an out-neighbor in $U$ (resp. $W$).
    Then, there exists a 1-friendly partition in $G$ with one part containing $U$
    and the other part containing $W$.  
\end{lemma}

For $r>1$ the existence of non-trivial $r$-friendly partitions remain open:

\begin{question}[see \cite{Alon06}]\label{q:r_friendly_partition}
    Let $r>1$ be an integer, does there exist an integer $d=d(r)$
    such that in every digraph $G$ whose minimal out-degree is at least $d(r)$
    there exists an $r$-friendly partition?
\end{question}

We make a brief remark regarding this question in Section~\ref{sec:appendix}
    (see Proposition~\ref{pro:ron}).

Unless stated otherwise, for the remainder of this manuscript we focus on the case $r=1$,
    and refer to $1$-friendly partitions simply by friendly partitions.
    As discussed above, in order to guarantee the existence of 
    (non-trivial) friendly partitions, $d$ must be at least $3$. 

Let us go back to the setting of assigning students for classes.
    {\it In this context, it would be useful to have many friendly partitions 
    with diverse properties that can be efficiently found.}
    This raises a host of questions:
    {\it How many friendly partitions are there:}
    what is the minimal number $t(d,n)$ of friendly partitions a digraph $G$
    with $n$ vertices and out-degree $d$ can have?
    By Theorem~\ref{thm:thomassen} and Lemma~\ref{lem:extending_friendly_sets}
    we have $t(3,n)\geq 1$ for all $n$.
    Is it the case that $\lim_{n\to\infty}t(3,n)=\infty$?
    Is it the case that $\lim_{n\to\infty}t(d,n)=\infty$ for some fixed $d$?
    {\it How well does the family of friendly partitions separate the vertices:}
    can every pair of vertices be separated by some friendly partition 
    (provided that $d$ is a sufficiently large constant)? Can most pairs be separated?
    In the following subsections, we address these questions
    as well as other related questions in more detail.

\subsection{How Well Do Friendly Partitions Separate the Vertices?}
Imagine that there is a small group $S\subseteq V$ of $s=5$ students
    that can control the preferences of all other kids.
    Can the students in $S$ devise lists $L(u)$ for every student $u$
    such that the teachers will have to assign all students in $S$ together
    to the same class?
    This suggests the following definition:
    Let $G=(V,E)$ be a digraph, let $S\subseteq V$.
    We say that $S$ is separable if there exists a friendly partition $V=U\cup W$
    such that both $U\cap S\neq\emptyset$ and $W\cap S\neq \emptyset$.
    The above motivation question amounts to the following:
    
\begin{question}[Separability]\label{q:sep}
    Does there exist a choice of $d$ and $s$
    such that in any digraph $G=(V,E)$ with out-degree $d$,
    every $S\subseteq V$ of size $k$ is separable?
    How about the case where $d=3$  and $s=2$:
    can every pair of vertices be separated if all out-degrees are at least $3$?
\end{question}

We present the following partial result
    that it is sufficient to consider strongly connected digraphs;
    that is, digraphs in which there is an oriented path
    from every vertex to every other vertex.
    
\begin{proposition}\label{pro:strongly_connected}
    Let $d \geq 3$. If every strongly connected digraph
    with minimum out-degree at least $d$
    satisfies that each pair of vertices in it is separable,
    then every digraph with minimum out-degree at least $d$ satisfies this property.
\end{proposition}

While we do not know an answer to Question~\ref{q:sep},
    our main technical result yields that if $d\geq 15$ 
    then there exists a vertex which can be separated from all other vertices:
    
\begin{theorem}[A Separable Vertex Exists]\label{thm:separable_vertex}
    Let $G=(V,E)$ be a digraph with minimum out-degree at least 15.
    Then, there exists $u\in V$
    such that for all $w\in V\setminus\{u\}$ there exists a friendly partition
    that separates $u$ and $w$.
\end{theorem}

We prove this theorem in Section~\ref{sec:existence_of_a_separable_vertex}. 
    Our proof follows by contradiction
    by considering a counter-example $G$ of minimal size. 
    We show that since the minimum out-degree of $G$ is $\geq 15$,
    it must contain $3$ cycles $C_1,C_2,C_3$ such that $C_1$ and $C_2$ intersect,
    but $C_3$ is disjoint from them. Then, by exploiting the minimality of $G$, 
    we show that one of the vertices on $C_1\cup C_2$ is separable,
    which is a contradiction.
    The existence of such three cycles hinges on a result by Thomassen~\cite{Thomassen83}
    which asserts that every digraph with minimal out-degree at least 15
    contains 3 disjoint cycles.
    Thomassen also conjectures that a minimal degree of 5 is sufficient.
    If true, this would imply that one can improve Theorem~\ref{thm:separable_vertex}
    by replacing $15$ by $5$.

\subsection{How Many Friendly Partitions Are There?}
In the previous section we interpreted ``richness'' of a family of partitions 
    in terms of the ability to separate sets of vertices.
    An alternative, perhaps more simplistic, interpretation of ``richness''
    is obtained by counting: i.e.\ bigger families are richer.
    How large must the family of friendly partition in a digraph
    with out-degree $d\geq 3$ be?
    Theorem~\ref{thm:thomassen} implies that there is at least one such partition.
    Is this tight? Are there arbitrarily large digraphs with out-degree $3$
    all of which have only $O(1)$ many friendly partitions?
    In Section~\ref{sec:number_of_friendly_partition},
    we strengthen Theorem~\ref{thm:thomassen}
    and show that when $d=3$ there must be at least $2$ friendly partitions:

\begin{theorem}\label{thm:more_partitions}
    In every digraph with minimum out-degree at least $3$
    there are at least $2$ distinct friendly partitions.
\end{theorem}

\subsubsection{Counting versus Separating}\label{sec131}
The number of friendly partitions and their separation capabilities are linked.
    For example, assume every pair of vertices in a digraph $G=(V,E)$
    can be separated by a friendly partition,
    and let $k$ denote the number of friendly partitions in $G$. 
    We claim that $k\geq \log(n)$, where $n$ is the number of vertices.
    To see this, assign to every vertex $v$ a binary string $b_v$ of length $k$,
    such that $b_v(i)=0$ if and only if $v$ belongs to the left part
    of the $i$-th friendly partition.
    Since every pair of vertices $u,v$ are separable by some friendly partition,
    it follows that $b_v\neq b_u$. 
    Thus, all binary strings are different and $2^k \ge n$ as claimed.

The next theorem implies a statement in the opposite direction:
    if the number of friendly partitions tends to infinity with $n$
    then the family of friendly partitions must separate
    all subsets of some fixed size (independent of $n$).

Recall that $t(d,n)$ denotes the minimum number of friendly partitions that
    exist in any digraph with $n$ vertices and out-degree $d$,
    and that a subset $S\subseteq V$ is called separable
    if there exists a friendly partition $V=V_1\cup V_2$
    such that both $V_1\cap S\neq\emptyset$
    and $V_2\cap S\neq\emptyset$.

\begin{theorem}~\label{thm:dichotomy}
    Fix $d \in \N$ then the following statements are equivalent:
    \begin{enumerate}
        \item There exists $s \in \N$ such that in every digraph
            with minimum out-degree at least $d$, 
            every subset of $s$ vertices is separable.
        \item In every digraph with minimum out-degree at least $d$,
            every subset of $s=d-1$ vertices is separable.
        \item The function $t(d,n)$ is unbounded (as a function of $n$).
        \item $t(d,n) \geq \log(n) - \log(d-2)$.
    \end{enumerate}
\end{theorem}

Theorem~\ref{thm:dichotomy} is proved in Section~\ref{sec:number_of_friendly_partition}.
    Note that this theorem implies
    that either $t(3,n)$ is bounded by a constant for every $n$,
    or else it must be the case that each pair of vertices must be separable.
    Thus, to prove that each pair of vertices is separable
    in every digraph with out-degree $3$,
    it suffices to show that the number of friendly partitions is unbounded.

Another corollary of Theorem~\ref{thm:dichotomy} is a dichotomy for $t(d,n)$: 
    for every fixed $d$, the function $t(d,n)$ is either upper bounded by a constant, 
    or it tends to infinity at a rate of at least $\Theta(\log n)$. 
    (E.g., it is impossible that $t(d,n) \in \Theta(\sqrt{\log n}$).)

The equivalence between items $1$ and $2$ implies
    that if there exists some $s\in\mathbb{N}$
    such that every subset of $s$ vertices is separable,
    then this already holds for $s=d-1$.
    The next result asserts that in this case it holds
    that every pair of vertices is separable in every digraph
    with minimum out-degree at least $d+1$:
\begin{theorem}~\label{thm:separability_for_d+1}
    Let $d\in\mathbb{N}$. If $t(d,n)$ is unbounded as a function of $n$,
    then in every digraph with minimum out-degree at least $d+1$,
    each pair of vertices is separable.
\end{theorem}

\subsection{Special Digraphs}
It is natural to explore the above questions under additional assumptions
    on the preference digraph.
    For example, it seems reasonable to assume
    that the in-degrees of the vertices are not too large.
    (I.e.\ that each student is listed by a bounded number of their school-mates.)
    Under such an assumption, a standard application of Lovász local lemma%
    \footnote{Lovász local lemma (symmetric version, see~\cite{Lov73}):
    Let $A_1, \ldots, A_n$ be events
    such that each $A_i$ occurs with the probability at most $p$
    and each $A_i$ is independent of all but at most $d$ other $A_j$.
    Then, $\Pr\left[\cap_{i=1}^n \overline{A_i}\right] > 0$
    provided that $e\cdot p\cdot (d+1) \leq 1$.}
    yields an affirmative answer to Question~\ref{q:sep}:
    
\begin{theorem}\label{thm:lll_c}
    Let $G=(V,E)$ be a directed graph with minimum out-degree $\geq d$
    and maximum in-degree $\delta\leq \frac{2^{d-1} - e}{ed}$. 
    Then, for every pair of distinct vertices $v_1,v_2\in V$
    there exists a friendly partition which separates it.
\end{theorem}
\begin{proof}
    It will be convenient to assume that all out-degrees are exactly $d$. 
    This is without loss of generality
    because we can always remove edges until this is satisfied.
    (Notice that the maximum in-degree can not increase when removing edges.)

Let $v_1,v_2\in V$ be distinct vertices.
    Draw a random partition of $V$ into two parts $V_1,V_2$
    which separates $v_1,v_2$ as follows: 
    $v_1$ is assigned to $V_1$, $v_2$ is assigned to $V_2$,
    and the part of every other vertex is chosen independently with probability
    $\frac{1}{2}$. For $w\in V$, let $A_w$ denote the event
    that every out-neighbor of $w$ is on a different part than $w$.  
    Notice that the probability of each $A_w$ is at most
    $\left(\frac{1}{2} \right)^{d-1}$.

We prove that with a positive probability such a random partition is friendly.
    Note that a partition is friendly
    if and only if it does not belong to any of the events $A_w$.
    Thus, it suffices to prove that with a positive probability
    none of the events $A_w$ holds.
    Towards this end we use the Lovász Local Lemma: 
    fix $w\in V$ and note that $A_w$ is independent of all events $A_u$
    such that $u$ is not an out-neighbor of $w$ 
    and $w$ and $u$ have no common out-neighbor;
    there are~$d + d(\delta - 1) = d\delta$ such events $A_u$.
    Thus, every $A_w$ is independent of all but at most $d\delta$ events $A_u$.
    Therefore, by the Lovász Local Lemma there exists a friendly partition
    separating~$v_1$ and~$v_2$, provided that
    \[e\left(\frac{1}{2}\right)^{d-1}\left(\delta d + 1\right) \leq 1,\] 
    which is equivalent to $\delta \leq \frac{2^{d-1} - e}{ed}$.
\end{proof}

Similarly, using the multiplicative form of Chernoff bound\footnote{
    (Multiplicative) Chernoff bound (see~\cite{Alon00}):
    Let $n\in\mathbb{N}, p\in[0,1]$ and let $B=B(n,p)$ denote a binomial random variable.
    Then, $\Pr\left[B\leq (1-\varepsilon)\mu)\right]
    \leq \exp\bigl(-{\frac{\varepsilon^2\mu}{2}}\bigr)$,
    for every $\varepsilon > 0$ where $\mu = np$ is the expectation of $B$.}
    one can separate each pair of vertices by an $r$-friendly-partition
    for $ r = (1 - \varepsilon)\frac{d-1}{2}+1$
    if $d$ is sufficiently large and the maximum  in-degree is bounded. 
    E.g. for $\varepsilon = \frac{1}{2}$ we get the following.

\begin{theorem}
    Let $G=(V,E)$ be a directed graph with minimum out-degree $\geq d$
    and maximum in-degree $\leq e^{\frac{d-1}{16} - 1}{d}$. For every pair of distinct
    vertices $v_1, v_2 \in V$ there exists an $r$-friendly partition which separates it
    for $r = \frac{d-1}{4} + 1$.
\end{theorem}
\begin{proof}
    The proof is essentially same as the proof of Theorem~\ref{thm:lll_c}.
    Again, without lost of generality we assume that all out-degrees are exactly $d$.
    We choose $v_1, v_2 \in V$ and draw a random partition
    of $V$ into $V_1$ and $V_2$ such that
    $v_1 \in V_1$ and $v_2 \in V_2$ and the part for of every other vertices
    is chose independently with probability $\frac{1}{2}$.
    Now, we need to bound the probability of an event $A_w$ representing that
    there are $\leq \frac{d-1}{4}$ out-neighbors of $w$ being in the same part as $w$.
    This can be bounded by $e^{-\frac{d-1}{16}}$ by Chernoff bound.
    As in the previous case, each event $A_w$ is independent of all
    but $d\delta$ events $A_u$.
    Therefore, if $\delta \leq e^{\frac{d-1}{16} - 1}{d}$
    Lovász Local Lemma guarantees there is
    a $\left(\frac{d-1}{4} + 1\right)$-friendly partition separating $v_1$ and $v_2$.
\end{proof}

\paragraph{Vertex-Transitive Digraphs.}

Theorem~\ref{thm:separable_vertex}, which asserts the existence of a vertex
    which is separable from all other vertices in digraphs
    with minimum out-degree at least $15$,
    also answers Question~\ref{q:sep} in the affirmative
    for {\it vertex-transitive} digraphs with the same minimum out-degree. 
    Indeed, recall that vertex-transitive digraphs are digraphs
    such that for every pair of vertices $u,v$
    there exists an automorphism of the digraph $\tau$ such that $\tau(u)=v$.
    Now, since for every friendly partition $\{V_1,V_2\}$ of $V$ it holds
    that also $\{\tau(V_1),\tau(V_2)\}$ is a friendly partition,
    we get that the existence of a single vertex
    which is separable from any other vertex implies that all vertices have this property.

However, for vertex-transitive digraphs we can use Theorem~\ref{thm:lll_c}
    to obtain better bounds.
    Indeed, Theorem~\ref{thm:lll_c} applies to any digraph
    in which all vertices have the same out-degree and in-degree. 
    In particular, regular digraphs with out-degree (and thus also with in-degree) equal 9
    fulfill the condition of Theorem~\ref{thm:lll_c} 
    and thus the following holds:

\begin{corollary}
    Let $G$ be a regular digraph with degree at least $d=9$.
    Then, every pair of vertices in $G$ is separable.
\end{corollary}
    
Can one further improve this bound for vertex-transitive digraphs: is $d=9$ tight?
    We conjecture that $d=3$ (which is clearly a lower bound,
    as witnessed by the triangle, see Figure~\ref{fig:K_3}),
    is the tight bound for vertex-transitive digraphs.
    In Section~\ref{sec:vertex-transitive_graphs},
    we prove some partial results towards proving this conjecture:
    
\begin{proposition}\label{pro:vertex_transitive_graph_prime}
    If the number of vertices in a vertex-transitive digraph
    with degree at least $3$ is prime
    then each pair of vertices is separable.
\end{proposition}

\begin{proposition}\label{pro:independent_sets}
    Let $\sim$ denote the following equivalence relation
    on the set of vertices of a digraph $G$.
    \begin{center}
        ``$u\sim v\quad \iff\quad $   $u$ cannot be separated from $v$
        by a friendly partition.''
    \end{center}
    Then, if $G$ is vertex-transitive with degree $d\geq 3$
    then each equivalence class of $\sim$ is an independent set. 
    Moreover, each vertex has its out-neighbors in at least $3$ different classes.
\end{proposition}

Curiously, Proposition~\ref{pro:independent_sets} implies that if there are two vertices
    which cannot be separable in a vertex transitive digraph of out-degree $d\geq 3$
    then there is no edge between them, which seems somewhat counter-intuitive.

\paragraph{Infinite Digraphs.}
It is also interesting to explore the properties of friendly partitions
    in infinite digraphs:
    in the last section (Section~\ref{sec:infinite_graphs}),
    we generalize the existence of friendly partitions to infinite digraphs
    with minimum out-degree 3.
    
\begin{proposition}\label{pro:infinite_graphs}
    Let $G$ be a (possibly infinite) digraph with minimum out-degree at least $3$.
    Then, there exists a friendly partition in $G$.
\end{proposition}

Perhaps surprisingly, the proof we found for Proposition~\ref{pro:infinite_graphs}
    is more complex then one might expect. 
    In particular, we could not find a reduction to the finite
    case using standard arguments such as compactness.

\subsection{Organization}
The article is organized as follows. In Section~\ref{sec:existence_of_a_separable_vertex},
    we prove our main result.
    That is, in each digraph of minimum out-degree at least 15
    there is a vertex separable from each other vertices
    (Theorem~\ref{thm:separable_vertex}).

In Section~\ref{sec:number_of_friendly_partition},
    we prove Theorem~\ref{thm:dichotomy} and Theorem~\ref{thm:separability_for_d+1}
    connecting the function $t(d,n)$ and separability of $k$-tuples.
    In this section we also prove the existence of at least 2 friendly partitions
    in each digraph of minimum out-degree $3$ (Theorem~\ref{thm:more_partitions}).

In Section~\ref{sec:strongly_connected_graphs}, we prove Proposition~\ref{pro:strongly_connected} showing 
that it suffices to only consider strongly regular graphs to show that any digraph with sufficiently large out-degrees,
satisfies that each pair of vertices in it is separable.

In Section~\ref{sec:vertex-transitive_graphs},
    we prove our partial result for vertex transitive digraphs.
    Namely, if there is a pair of unseparable vertices
    in vertex transitive digraph of minimum out-degree 3,
    then the number of its vertices is not a prime number
    (Proposition~\ref{pro:vertex_transitive_graph_prime})
    and such unseparable vertices are not connected by an oriented edge
    (Proposition~\ref{pro:independent_sets}).

In the last section, Section~\ref{sec:infinite_graphs},
    we prove the existence of friendly partition
    for infinite digraphs with minimum out-degree at least $3$
    (Proposition~\ref{pro:infinite_graphs}).

\section{Existence of a Separable Vertex}~\label{sec:existence_of_a_separable_vertex}%
In this section we prove Theorem~\ref{thm:separable_vertex}.
    We begin by recalling a useful definition introduced by Thomassen
    (\cite{Thomassen83}).
    
\begin{definition}
    Let $G=(V,E)$ be a digraph. An edge $u\to v\in E$ is called \underline{dominated}
    if $u,v$ have a common in-neighbor. (See Figure~\ref{fig:dominated edge}.)
\end{definition}

\begin{figure}
    \centering
    \includegraphics{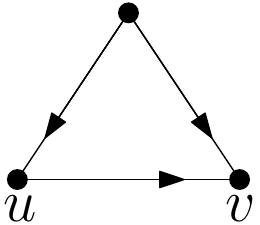}
    \caption{\small A dominated edge $uv$.}
    \label{fig:dominated edge}
\end{figure}

Thomassen used a process which, given a digraph $G=(V,E)$, produces
    a ``compressed'' digraph $G'=(V',E')$ such that every edge $e'\in E'$
    is either dominated, or a part of a 2-cycle.
    The digraph $G=(V',E')$ is produced as follows:
    as long as there exists an edge $u\to v$ which is not dominated
    nor a part of 2-cycle, 
    delete all the edges going out from $u$ and identify vertices $u$ and $v$.
    Let us call this procedure \emph{contraction of non-dominated edge}.
    See Figure~\ref{fig:contracting}.
    Crucially, notice that every friendly partition of $G'$ is naturally ``decompressed'' to
    a friendly partition of $G$ by putting identified vertices in the same part. 
\begin{figure}
    \centering
    \includegraphics[scale=.8]{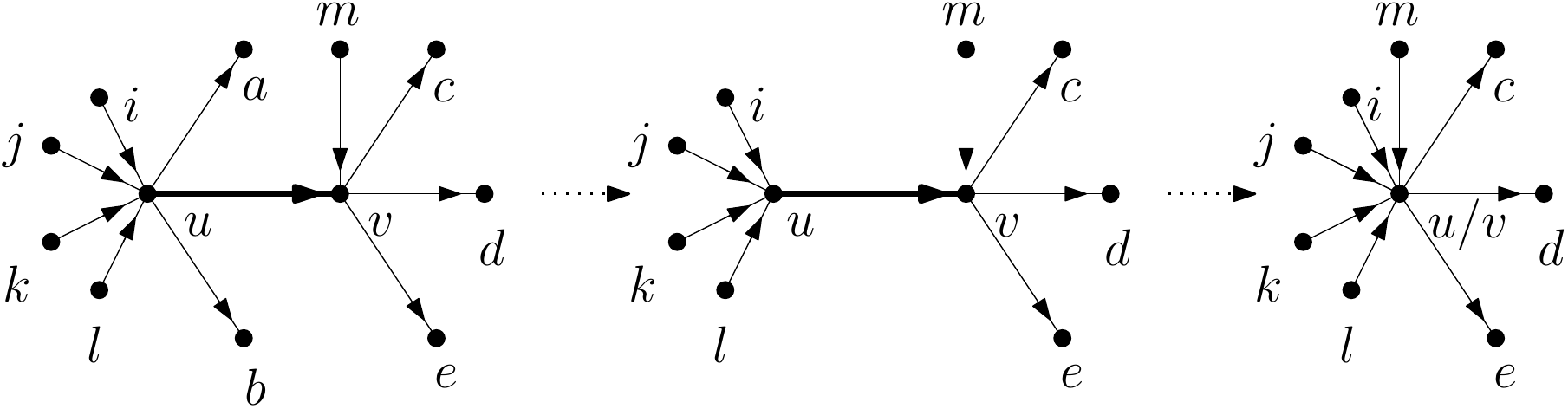}
    \caption{\small A picture showing how to contract a non-dominated edge $uv$
    which is not part of a two-cycle while preserving the minimum out-degree in the graph.}
        \label{fig:contracting}
\end{figure}
The compressed digraph $G'$ satisfies the following useful property.

\begin{lemma}[\cite{Thomassen83}]\label{lem:thomassen_inneighborhood_cycle}
    Let $G=(V,E)$ be a digraph such that each edge in it
    is either dominated or a part of a 2-cycle,
    and let $v\in V$ be a vertex which is not part of a 2-cycle.
    Then, there exists a cycle in the in-neighborhood of $v$.
\end{lemma}

Indeed, this lemma follows because every vertex $u$ in the digraph
    induced by the in-neighborhood of $v$ has a positive in-degree
    (because the edge $u\to v$ is dominated).

Another theorem due to Thomassen which will be useful in our proof is the following.
\begin{theorem}[\cite{Thomassen83}]\label{thm:thomassen_3cycles}
    Each digraph with minimum out-degree 15 contains three disjoint cycles.
\end{theorem}

Now, we continue with three lemmas which,
    combined with Theorem~\ref{thm:thomassen_3cycles},
    imply Theorem~\ref{thm:separable_vertex}. 
\begin{lemma}\label{lem:two_intersecting_cycles}
    If $G$ is a digraph with minimum out-degree at least 2
    then it contains two intersecting cycles.
\end{lemma}
\begin{proof}
    Let $\mathcal{C}=\{C_1, \ldots, C_k\}$ be a maximal family of disjoint cycles. 
    Note that $\mathcal{C}\neq\emptyset$ since the minimum out-degree least 2
    and thus, $G$ contains a cycle.
    Now, let us delete all edges of $C_1, \ldots, C_k$
    and denote the resulting digraph as $G^\prime$.
    Since $C_1, \ldots, C_k$ are disjoint,
    the minimum out-degree of $G^\prime$ is at least 1. 
    Therefore, $G^\prime$ contains a cycle
    which intersect at least one cycle from $\mathcal{C}$ by maximality of $\mathcal{C}$.
\end{proof}

\begin{lemma}\label{lem:intersecting_cycles_and_cycle_exist}
    Let $d$ be the minimum integer such that every digraph
    with minimal out-degree $\geq d$
    must contain three disjoint cycles.\footnote{By Theorem~\ref{thm:thomassen_3cycles},
    $d\leq 15$.}
    Then, $G$ contains two intersecting cycles and one cycle disjoint from them.
\end{lemma}

Before we prove this lemma, note that $d$ must be $>4$ since,
    for example, the complete directed graph on 5 vertices has minimum out-degree $4$
    but cannot contain three disjoint cycles. See Figure~\ref{fig:k5}.
    
\begin{figure}
    \centering
    \includegraphics[scale=.8]{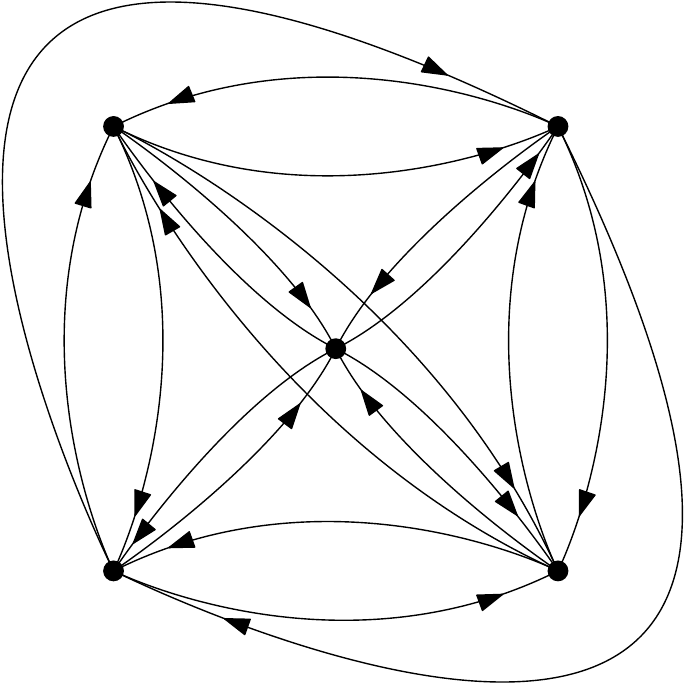}
    \caption{\small The complete directed graph on 5 vertices.}
    \label{fig:k5}
\end{figure}

\begin{proof}[Proof of Lemma~\ref{lem:intersecting_cycles_and_cycle_exist}]
    We prove by contradiction: let $G = (V, E)$ be a smallest counterexample.
    Every edge of $G$ must be dominated or a part of a 2-cycle,
    otherwise we can contract this edge and get a smaller counterexample.%
    \footnote{Indeed, two intersecting cycles and a third cycle disjoint
    from them in the contracted digraph are naturally lifted to three cycles
    with the same properties in the original digraph.}
    If $G$ contains a 2-cycle $C$ then $G$
    without $C$ has minimum out-degree at least 2.
    Therefore, by Lemma~\ref{lem:two_intersecting_cycles}
    it contains two intersecting cycles.
    Such cycles are disjoint from $C$ in $G$;
    a contradiction.
    
Thus, $G$ does not contain a $2$-cycle and each edge in $E$ is dominated.
    Therefore, by Lemma~\ref{lem:thomassen_inneighborhood_cycle}, 
    the in-neighborhood of each vertex $v\in V$ contains a cycle. 

Let $I_v$ denote a cycle in the in-neighborhood of vertex $v$,
    let $\mathcal{C}=\{C_1, \dots , C_k\}$ be a maximal family of disjoint cycles 
    (by the assumption on $d$ we know that $k \ge 3$),
    and let $V_C = \bigcup_{i=1}^k V(C_i)$ be the set of all their vertices.
    By maximality of $\mathcal{C}$,
    for each $v \in V$ the cycle $I_v$ must intersect some $C_i\in\mathcal{C}$.
    Consider the following cases.
\begin{itemize}
    \item Case 1: The in-neighborhood cycle $I_v$ of every vertex $v$ from $V_C$
        is $C_i$ for some $i \in \{1, \dots, k\}$. 
        Define a digraph $H$ whose vertices are cycles $C_i$
        and its edges are all the edges $C_i\to C_j$ such
        that there exists $v \in C_j$ such that $I_v = C_i$.
        By the assumption, every vertex in $H$ has in-degree at least one.
        Thus $H$ has to contain a cycle. If $H$ is a cycle
        then every two neighboring $C_i,C_j$ in it are {\it fully connected} in $G$:
        that is, for each vertex $v$ of $C_i$
        there is an edge to \underline{every} vertex of the following cycle $C_j$. 
        (Indeed, else the in-degree of $C_j$ in $H$ would be at least $2$,
        contradicting the assumption that $H$ is a cycle.)
        Now, since each $C_i$ contains at least 3 vertices,
        one can find two intersecting cycles and a cycle disjoint from them as well.
        (See Figure~\ref{fig:fully_connected_cycles}.)
        \begin{figure}
            \centering
            \includegraphics[scale=.8]{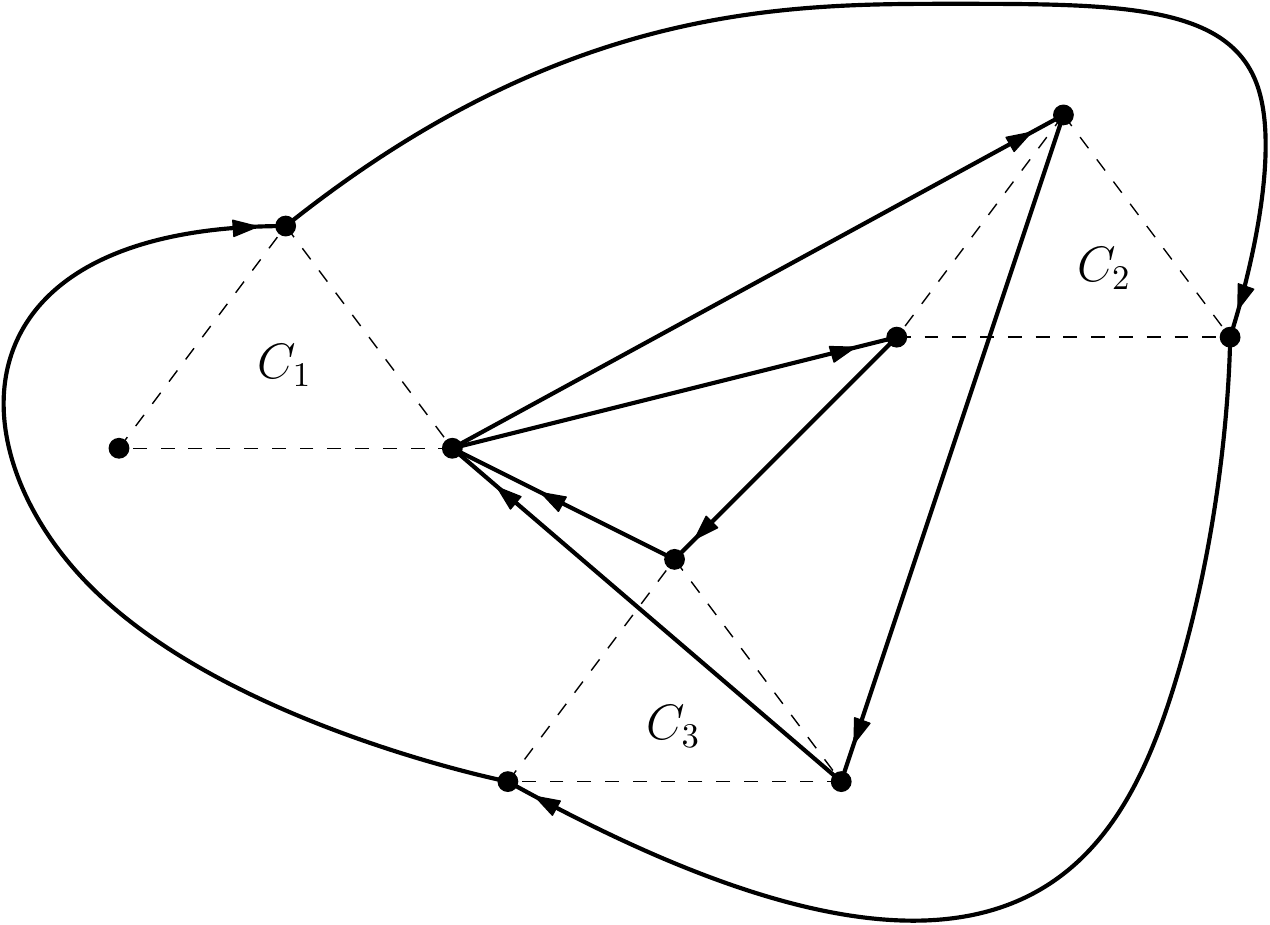}
            \caption{\small A picture showing how to find two intersecting cycles
                and one another cycle disjoint from them
                in the proof of Lemma~\ref{lem:intersecting_cycles_and_cycle_exist}
                in the case where $H$ is a cycle.}
                \label{fig:fully_connected_cycles}
        \end{figure}
        
    Else, $H$ is not a cycle and hence, there must exist a cycle in $H$
        which does not contain all vertices in $H$.
        In such a case one can also find two intersecting cycles
        and a cycle disjoint from them. (See Figure~\ref{fig:H_is_not_a_cycle}.)
        \begin{figure}
            \centering
            \includegraphics{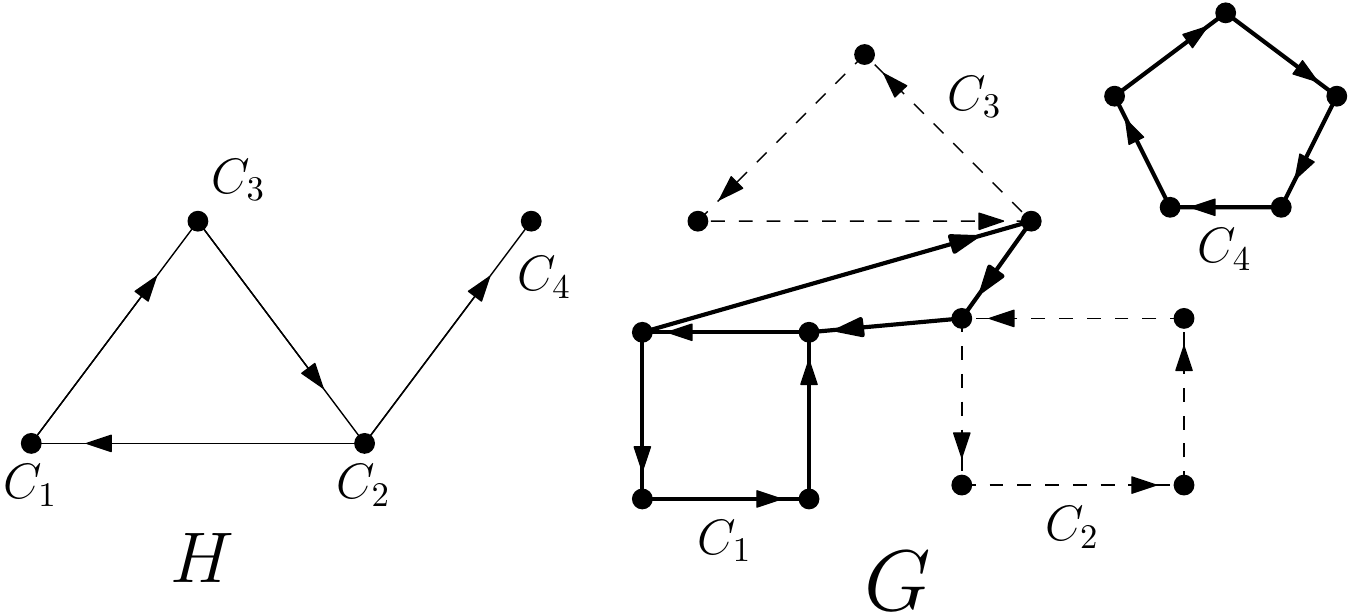}
            \caption{\small A picture showing how we can choose two intersecting cycles
                and one another cycle disjoint from them
                in the proof of Lemma~\ref{lem:intersecting_cycles_and_cycle_exist}
                in the case when there is a cycle in $H$
                which does not contain all of its vertices.}
            \label{fig:H_is_not_a_cycle}
        \end{figure}
        
    \item Case 2: There exists $v \in V_C$ such that $I_v \neq C_i$ for all $i\leq k$. 
        Without loss of generality, assume $v\in V(C_1)$. 
        If there exists $C_i$ for some $i \in \{1, \dots, k\}$
        such that $I_v$ and $C_i$ are disjoint
        then we can find the two intersecting cycles
        and a cycle disjoint from them as follows:
        by maximality of $\mathcal{C}$, $I_v$ intersects $C_j$ for some $j \neq i$,
        and the cycle $C_i$ is disjoint from both of them.
        
        We are left with the case that $I_v$ intersects every $C_i$. 
        Here, we can find the two intersecting cycles
        and a cycle disjoint from them as follows
        (see Figure~\ref{fig:I_v_intersect_every_C_i}):
        let $u \in I_v \cap C_1$ and let $w \in I_v$ be the first vertex
        in cycle $I_v$ after $u$.
        Then the set $V(C_1) \cup \{w\}$ induces two intersecting cycles.
        For the disjoint cycle, pick any $C_l, l \neq 1$ such that $w \not\in C_l$.
        (Such $C_l$ exists because $k\geq 3$.)
        \begin{figure}
            \centering
            \includegraphics{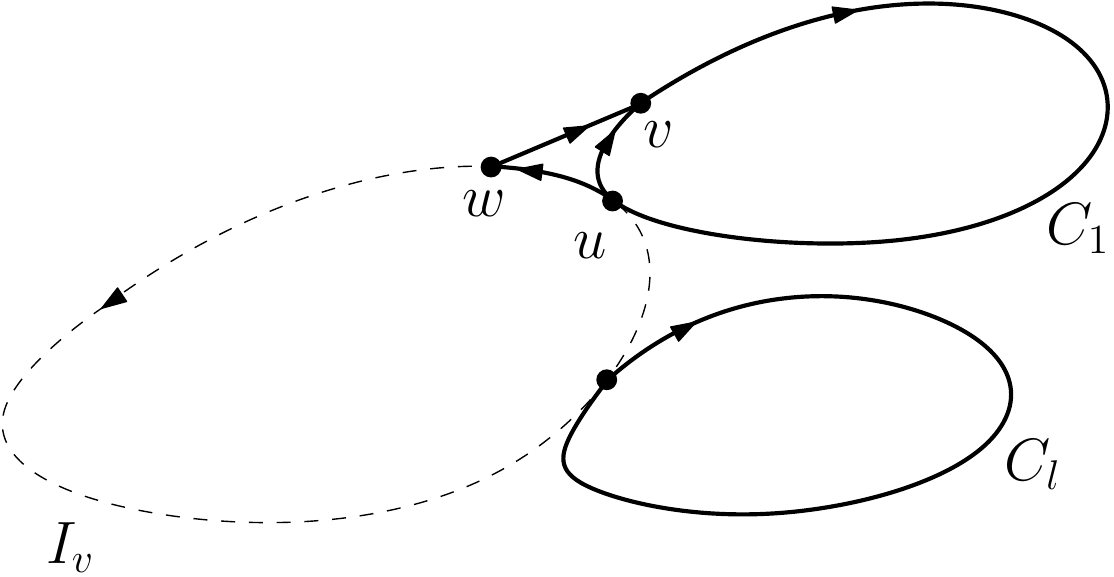}
            \caption{\small The two intersecting cycles induced by $V(C_1) \cup w$
                and the cycle $C_l$ disjoint from them
                from the proof of Lemma~\ref{lem:intersecting_cycles_and_cycle_exist}.}
                \label{fig:I_v_intersect_every_C_i}
        \end{figure}
    \end{itemize}
    This is in contradiction to $G$ being a counterexample.
\end{proof}

\begin{lemma}\label{lem:intersecting_cycles_and_cycle_imply_singleton}
    Let $d\geq 3$ be a sufficiently large integer
    such that every digraph with minimal out-degree $\geq d$
    must contain two intersecting cycles and another cycle disjoint from them.%
    \footnote{By Lemma~\ref{lem:intersecting_cycles_and_cycle_exist}, $d\leq 15$.}
    Then, every digraph $G$ with minimal out-degree $\geq d$ contains a vertex
    which is separable from every other vertex.
\end{lemma}
\begin{proof}
    We prove by contradiction: let $G = (V, E)$ be a smallest counterexample.
    Consider three cycles $C_1,C_2,C_3$ as guaranteed by the premise:
    let $I$ denote the set of vertices in the two intersecting cycles $C_1,C_2$
    and let $J$ denote the set of vertices of the remaining cycle $C_3$. 
    (Thus, $I\cap J=\emptyset$.)

Let $\{V_1, V_2\}$ be a friendly partition of $V$
    such that $I \subset V_1, J \subset V_2$ and $\lvert V_1\rvert$ is minimal.
    The idea of the proof is to find a vertex $v\in V_1$
    such that the partition $\{V_1\setminus\{v\},V_2\cup\{v\}\}$ is also friendly.
    Note that this implies that $v$ is separable and yields the desired contradiction.
    Note that $\{V_1\setminus\{v\},V_2\cup\{v\}\}$ is friendly if and only if
    $v$ has at least one out-neighbor in $V_2$
    and every in-neighbor of $v$ in $V_1$
    has at least one additional out-neighbor in $V_1$.

By the minimality of $\{V_1,V_2\}$ it follows 
    that if $v \in V_1\setminus I$ then all out-neighbors of $v$ are in~$V_1$. 
    Notice that at least one vertex $v_0\in I\subseteq V_1$ has an out-neighbor in $V_2$,
    because otherwise the digraph induced by $V_1$ would be a smaller counterexample: 
    indeed, if there exists $v\in V_1$ which is separable from any other vertex $u\in V_1$
    then $v$ is also separable from any other vertex in $V$
    (because $\{V_1,V_2\}$ is friendly.)
    Next, since $v_0$ is not separable, there must exist $v_1\in V_1$ such that:
\begin{itemize}
    \item[(i)] $v_0$ is the \underline{only} out-neighbor of $v_1$ in $V_1$:
        because otherwise the partition $V = \{V_1 \setminus \{v_0\}, V_2\cup \{v_0\}\}$
        is also friendly,
        which implies that $v_0$ is a separable vertex
        and therefore, yields a contradiction.
    \item[(ii)] $v_1\to v_0$ is an edge of one of the cycles $C_1$ or $C_2$:
        indeed, all vertices in $V_1\setminus I$ have all their out-neighbors in $V_1$.
        Thus, $v_0\in I$; now, since $v_0$ is the only out-neighbor of~$v_1$ in~$V_1$,
        it follows that $v_1\to v_0$ is a cycle-edge.
\end{itemize}
Similarly, there must exist $v_2\in V_1$ with analogous properties
    (i.e.\ $v_1$ is the only out-neighbor of $v_2$ in $V_1$,
    and $v_2\to v_1$ is an edge of one the cycles $C_1$ or $C_2$). 
    Continuing in this way we construct an endless sequence 
    $v_0, v_1, v_2, \dots\in V_1$ such that for each $i\geq 0$,
    $v_i$ is the unique out-neighbor of $v_{i+1}$ in $V_1$,
    and $v_{i+1}\to v_{i}$ is an edge on $C_1$ or $C_2$.
    In particular, at some point $i>0$, 
    we must encounter a vertex $v_i$ which is in the intersection
    of the cycles $C_1, C_2$.
    This is a contradiction
    because such a vertex has at least 2 out-neighbors in $I\subseteq V_1$.
    (See Figure~\ref{fig:separable_vertex}.)
\end{proof}
\begin{figure}
    \centering
    \includegraphics[scale=1]{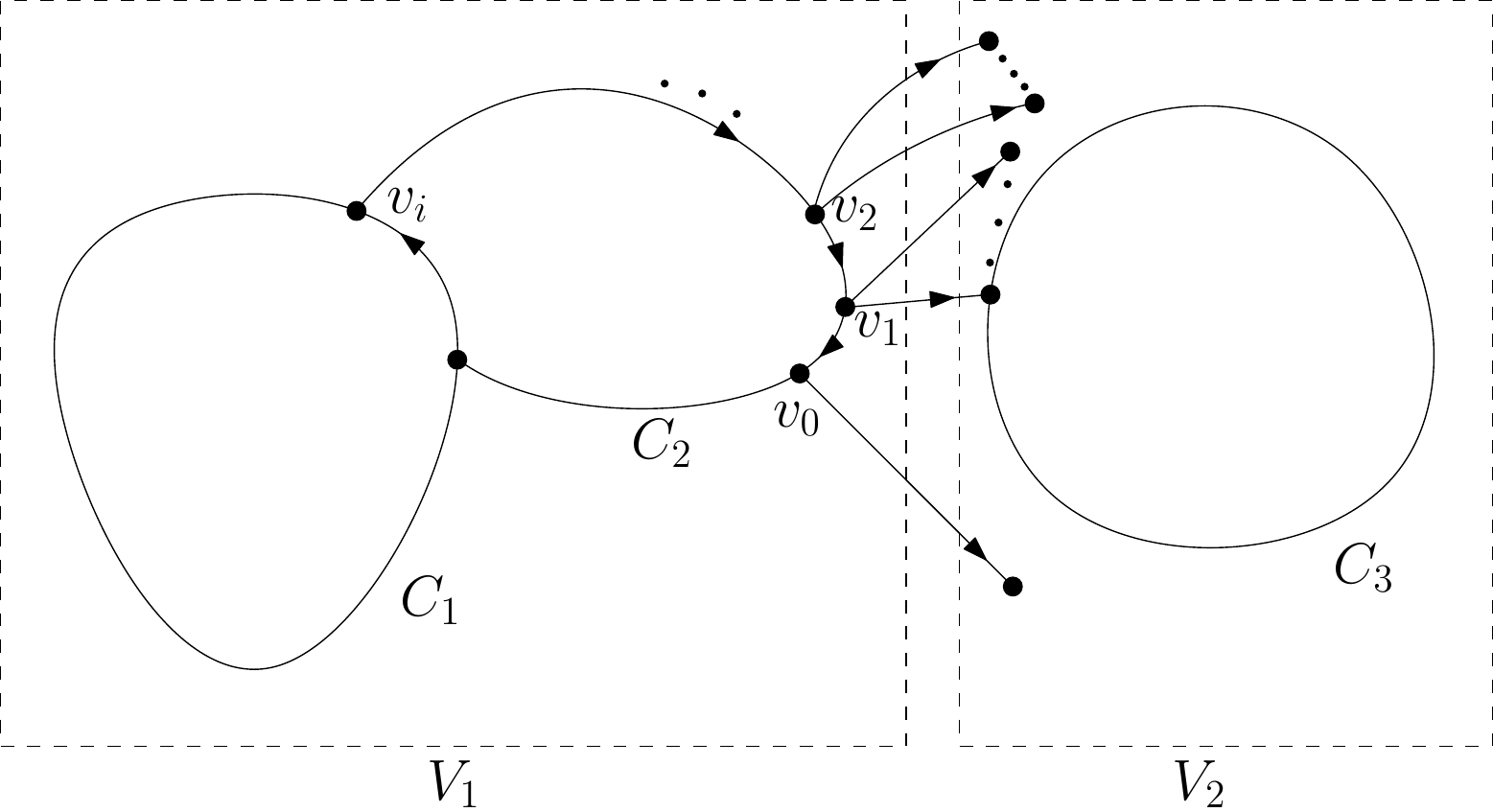}
    \caption{\small A picture showing the sequence of vertices $v_1, v_2, \ldots$
        from the proof of Lemma~\ref{lem:intersecting_cycles_and_cycle_imply_singleton}.}
        \label{fig:separable_vertex}
\end{figure}

Theorem~\ref{thm:separable_vertex} also implies the following.

\begin{corollary}
    Let $G$ be a digraph with minimum out-degree at least 15+$k$,
    then $G$ contains at least $k$ vertices separable from all the other vertices.
\end{corollary}
\begin{proof}
   It easily follows from Theorem~\ref{thm:separable_vertex}.
   We simply delete the vertex separable from all other vertices
   and apply the corollary again.
\end{proof}

\section{The Number of Friendly Partitions}\label{sec:number_of_friendly_partition}
In this section we prove Theorems~\ref{thm:more_partitions},~\ref{thm:dichotomy},
and \ref{thm:separability_for_d+1}.
We begin by proving Theorem~\ref{thm:dichotomy}.

\begin{reptheorem}{thm:dichotomy}[Restatement]
Fix $d \in \N$ then the following statements are equivalent:
    \begin{enumerate}
        \item There exists $s \in \N$ such that in every digraph
            with minimum out-degree at least $d$, 
            every subset of $s$ vertices is separable.
        \item In every digraph with minimum out-degree at least $d$,
            every subset of $s=d-1$ vertices is separable.
        \item The function\footnote{Recall that $t(d,n)$ denotes the minimum number of 
            friendly partitions that exist in any digraph with $n$ vertices
            and out-degree $d$}
            $t(d,n)$ is unbounded (as a function of $n$).
        \item $t(d,n) \geq \log(n) - \log(d-2)$.
    \end{enumerate}
\end{reptheorem}

\begin{proof}
The implication ``$4 \implies 3$'' is trivial.
    Therefore, it is sufficient to show ``$1 \implies 2$'',
    ``$2\implies 4$'', and ``$3 \implies 1$''
    since then we get
    \[1 \implies 2 \implies 4 \implies 3 \implies 1.\]

$1 \implies 2$: 
    let $G=(V,E)$ be a digraph with minimum out-degree $d$ and let $D_- \subset V$
    be a $(d-1)-$tuple of vertices we want to separate.
    We modify $G$ by adding a cycle $S$ to it of size $s$ and a directed edge
    from each vertex in $S$ to each vertex in $D_-$.
    See Figure~\ref{fig:d-1_separable}.
\begin{figure}
    \centering
    \includegraphics[scale=1.3]{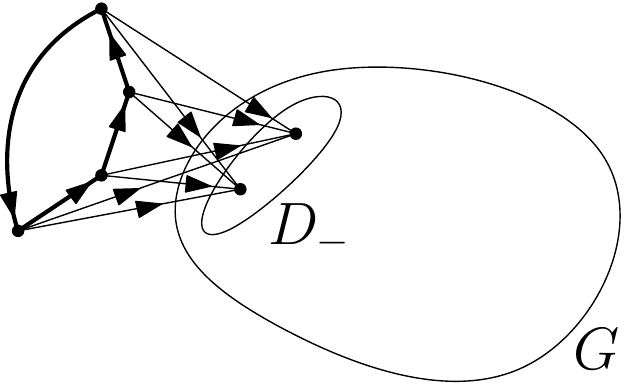}
    \caption{\small A picture of the digraph $G$ together with the cycle of size $s$
        showing how we proceed in the proof of Theorem~\ref{thm:dichotomy}.
        In this example, we suppose $d = 3$ and every 4-tuple is separable.}
        \label{fig:d-1_separable}
\end{figure}
Note that the minimum out-degree in the modified digraph remains $d$.
    Thus, by the assumption in Item $1$, there exists a friendly partition $\{W,U\}$
    that separates $S$. 
    We claim that this partition must separate $D_-$:
    indeed, otherwise all vertices in $D_-$ are in the same part, say $D_-\subseteq W$.
    However, there must be $s_0\in S$ such that $s_0\in U$; thus, since $S$ is a cycle,
    the unique out-neighbor of $s_0$ in $S$, denoted by $s_1$, must also be in $U$.
    Similarly, the out-neighbor of $s_1$ in $S$ must also be in $U$.
    Continuing in this way, we conclude that all vertices of $S$ are in $U$
    which contradicts the assumption that $S$ is separated by $\{W,U\}$.
    Thus, $D_-$ must be separable by $\{U,V\}$ as required.

\medskip

$2\implies 4$:
    assume every subset of $d-1$ vertices in a digraph $G=(V,E)$
    can be separated by a friendly partition,
    and let $k$ denote the number of friendly partitions in $G$. 
    We claim that $k\geq \log(n) - \log(d-2)$, where $n$ is the number of vertices.
    To see this, assign to every vertex $v$ a binary string $b_v$ of length $k$,
    such that $b_v(i)=0$ if and only if $v$ belongs to the left part
    of the $i$-th friendly partition.
    Note that $b_v=b_u$ if and only if $v$ and $u$ cannot be separated
    by a friendly partition.
    In other words, $b_u=b_v$ if and only if $u,v$ belong to the same equivalence class
    under the relation 
    \begin{center}
        ``$x\sim y\quad \iff\quad $ $x$ cannot be separated from $y$
        by a friendly partition.''
    \end{center}
    By the assumption in Item 2, each equivalence class is of size $\leq d-2$,
    and therefore, the number of equivalence classes is at least $\frac{n}{d-2}$.
    Hence, also the number of distinct binary vectors in the set $\{b_v: v\in V\}$
    is at least $\frac{n}{d-2}$,
    and since each vector $b_v$ has length $k$, it follows
    that $k\geq \log(\frac{n}{d-2})$ as required.

\medskip

$3 \implies 1$:
    we prove the contrapositive $\neg 1 \implies \neg 3$.
    Assume that for each $s \geq 2$
    there is a digraph with minimum out-degree at least $d$
    containing an unseparable $s-$tuple. 
    Pick such digraph $G=(V,E)$ for $s=d$
    and denote the unseparable $d$-tuple by $D$.

For every $n > \lvert V\rvert$ consider the digraph $G_n$ on $n$ vertices
    which is obtained by adding $n-\lvert V\rvert$ vertices to $G$
    and connect them with an out-going edge only to $D$
    (i.e.\ the new vertices have in-degree $0$).
    The digraph $G_n$ has minimum out-degree $d$ as well
    and has the same number of friendly partitions as $G$
    since each new vertex has to be in the same part as $D$.
    The digraphs $G_n$ witness that $t(d,n)$ is bounded.

This completes the proof of Theorem~\ref{thm:dichotomy}.
\end{proof}

\begin{reptheorem}{thm:separability_for_d+1}[{Restatement}]
    Let $d\in\mathbb{N}$.
    If $t(d,n)$ is unbounded as a function of $n$,
    then in every digraph with minimum out-degree at least $d+1$,
    each pair of vertices is separable.
\end{reptheorem}

\begin{proof}
    By Theorem~\ref{thm:dichotomy}, if $t(d,n)$ {is unbounded}
    then every $(d-1)$-tuple of vertices is separable in every digraph
    with minimum out-degree at least $d$.
    
Let $u,v$ be vertices in a directed graph $G$ with all out-degrees being at least $d+1$.
    By deleting $u$ we obtain a digraph $G^\prime$
    with all out-degrees being at least $d$. 
    Now, take a friendly partition of $G^\prime$ separating the
    out-neighborhood of $u$. (Such a partition is guaranteed to exist by the assumption.)
    Then we can extend such friendly partition of $G^\prime$ to a friendly partition of
    $G$ by adding $u$ to any  of the at least $2$ parts
    which contain an out-neighbor of $u$. 
    By picking a part which does \underline{not} contain $v$, we get a friendly partition
    that separates $u$ from $v$.
\end{proof}

In the rest of this section, we prove Theorem~\ref{thm:more_partitions}
    which asserts that every digraph with minimum out-degree at least $3$
    has at least $2$ friendly partitions.
    This generalizes Thomassen's result
    (Theorem~\ref{thm:thomassen} and Lemma~\ref{lem:extending_friendly_sets})
    which amounts to to the existence of at least $1$ friendly partition.

Recall from the beginning of Section~\ref{sec:existence_of_a_separable_vertex}
    that an edge is dominated if both its vertices have a common in-neighbor.
    Our strategy is at first to prove the desired result for digraphs
    in which each edge is dominated or a part of a 2-cycle.
    Then, we use Thomassen's reduction which implies this result to general digraphs.

We begin with a couple of lemmas which characterize the structure
    of digraphs with out-degree $3$ in which each edge is dominated 
    and there are no $2$-cycles.
    In a nutshell, these lemmas imply that such digraphs
    are closed under reversing the edges.
    (I.e.\ the digraph obtained by reversing all the edges
    is of the same kind.)

\begin{lemma}\label{lem:indegree}
    Let $G$ be a digraph without 2-cycles 
    such that all edges of $G$ are dominated and all vertices have out-degree 3.
    Then, all vertices in $G$ have in-degree 3.
\end{lemma}
\begin{proof}
    Since all vertices have out-degree 3, the average in-degree is 3.
    Assume towards contradiction that there is a vertex with in-degree greater than $3$.
    So, there has to be a vertex of in-degree less than $3$.
    However, by Lemma~\ref{lem:thomassen_inneighborhood_cycle}, 
    its in-neighbors form a cycle which has to be of length 2.
    This contradicts the assumption that $G$ does not contain $2$-cycles.
\end{proof}
    
\begin{lemma}\label{lem:common_vertex}
    Let $G$ be a digraph without 2-cycles 
    such that all edges of $G$ are dominated and all vertices have out-degree 3,
    and let $u\to v$ be an edge in $G$.
    Then, $u$ and $v$ have a common out-neighbor.
\end{lemma}

\begin{proof}
    By Lemma~\ref{lem:indegree}, also all the in-degrees in $G$ are 3. 
    Let $i,j$ denote the 2 other in-neighbors of $v$, apart from $u$,
    and let $o,p,q$ denote the 3 out-neighbors of $v$.
    (See Figure~\ref{fig:common_vertex}.)
    By assumption, each of the edges $v\to o,v \to p,v \to q$ is dominated,
    and the only vertices that can dominate these edges are the 3 in-neighbors $u,i,j$ of $v$.
    Since $u,i,j$ form a cycle in $G$ (as the in-neighborhood of $v$)
    and because the out-degrees in $G$ are $3$,
    it follows that each of $u,i,j$ dominates \emph{exactly}
    one edge from $v\to o, v\to p, v\to q$. 
    Thus, one edge, say $v\to o$, has to be dominated by $u$,
    and consequently, $o$ is a common out-neighbor of $u$ and $v$.
\end{proof}

\begin{figure}
    \centering
    \includegraphics{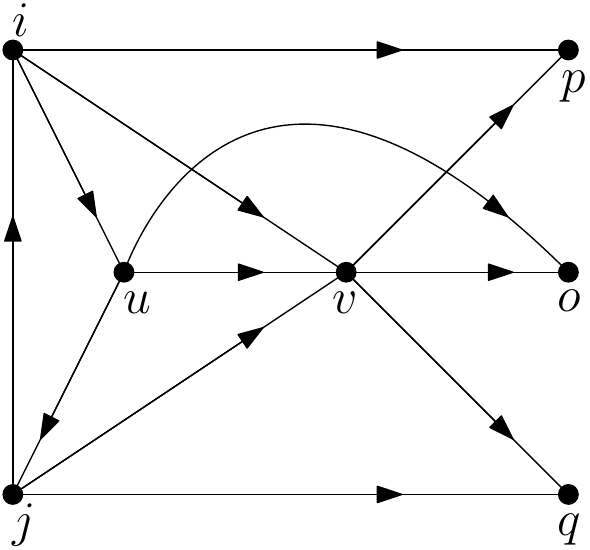}
    \caption{\small A neighborhood of the vertex $v$
        illustrating the situation of Lemma~\ref{lem:common_vertex}.}
        \label{fig:common_vertex}
\end{figure}

\begin{corollary}\label{cor:out_neighbors_cycle}
    Let $G$ be a digraph without 2-cycles 
    such that all edges of~$G$ are dominated and all vertices have out-degree 3,
    then the out-neighbors of each vertex in $G$ form a cycle.
\end{corollary}

\begin{proof}
    By reversing the orientation of each edge we obtain a digraph
    in which all vertices has out-degree 3 by Lemma~\ref{lem:indegree}
    and all edges are dominated by Lemma~\ref{lem:common_vertex}.
    Then, we can use Lemma~\ref{lem:thomassen_inneighborhood_cycle}.
\end{proof}

\begin{lemma}\label{lem:two_graphs}
    Up to isomorphism, there are exactly two non-empty digraphs $G$ satisfying the following properties:
    \begin{enumerate}
        \item Each vertex in $G$ has out-degree $3$.
        \item Every edge in $G$ is dominated.
        \item $G$ does not contain a $2$-cycle.
        \item The undirected graph underlying $G$ is connected.
    \end{enumerate}
\end{lemma}

\begin{proof}
    Consider such a digraph $G$ and pick any vertex $v$ in $G$.
    By Lemma~\ref{lem:indegree} also the in-degree of every vertex in $G$ is $3$.
    Let $w_1,w_2,w_3$ be $v$'s out-neighbors 
    and $u_1, u_2, u_3$ be $v$'s in-neighbors
    (see Figure~\ref{fig:two_cases_of_orientation}).

Lemma~\ref{lem:thomassen_inneighborhood_cycle}
    and Corollary~\ref{cor:out_neighbors_cycle} imply
    that each of the triplets $w_1,w_2,w_3$ and $u_1,u_2, u_3$
    form an oriented cycle in $G$.
    Let $C_{in}$ denote the cycle formed by $w_1, w_2, w_3$
    and $C_{out}$ denote the cycle formed by $u_1,u_2, u_3$.
    Without loss of generality, assume that $w_2$ is the out-neighbor of $w_1$ in $C_{in}$,
    and that $u_2$ is the out-neighbor of $u_1$ in $C_{out}$.
    In addition, each of the edges $v\to w_i$ is dominated by one of the $u_j$'s,
    and therefore, $G$ contains a matching of the form $w_i\to u_{\pi(i)}$
    for some permutation $\pi:[3]\to [3]$.
    Without loss of generality, assume that $\pi(1) = 1$;
    there are two cases:
\begin{itemize}
    \item $\pi(2) = 2$ (and therefore, $\pi(3) = 3$),
    \item $\pi(2) = 3$ (and therefore, $\pi(3) =2$).
\end{itemize}
In other words, with respect to the correspondence induced by $\pi$,
    either $C_{in}$ and~$C_{out}$ are oriented the same (when $\pi(2)=2$) or oppositely (when $\pi(2)=3$).
\begin{figure}
    \centering
    \includegraphics[scale=.95]{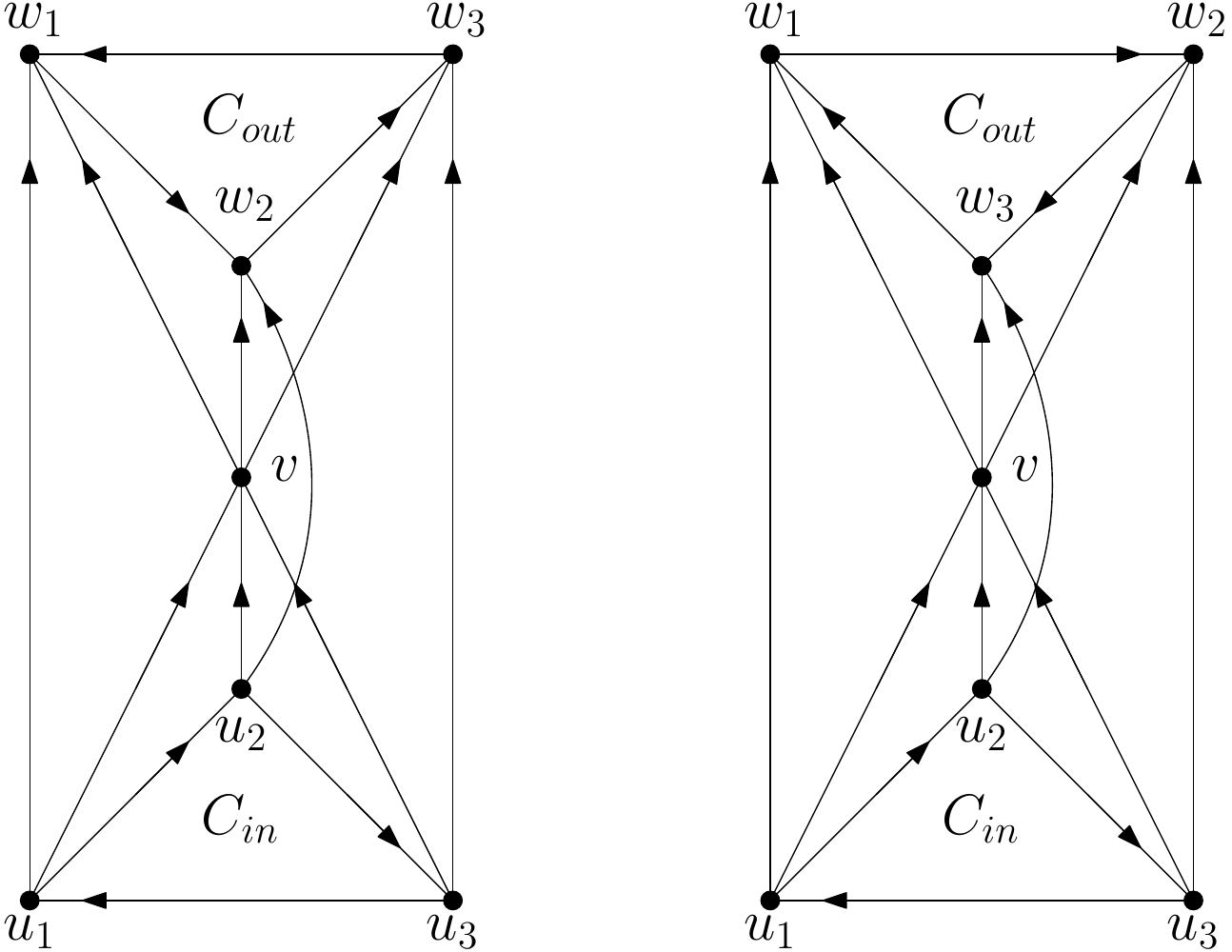}
    \caption{\small The two possible neighborhoods of the vertex $v$ 
        from the proof of Lemma~\ref{lem:two_graphs}.
        On the left, $C_{in}$ and $C_{out}$ are oriented the same ($\pi(2) = 2$),
        and on the right oppositely ($\pi(2) = 3$).}
        \label{fig:two_cases_of_orientation}
\end{figure}

By Corollary~\ref{cor:out_neighbors_cycle}
    the out-neighborhood of each vertex forms a cycle (triangle).
    One can verify that in each of the above cases there is a unique
    way of connecting edges from $C_{out}$ to $C_{in}$ 
    so that the out-neighborhoods of each vertex from $C_{in}$ form a cycle:
    for example $u_1\in C_{in}$ has out neighbors $v,u_2,w_1$,
    and since $u_2v$ and $vw_1$ exist as edges, it follows that
    $w_1u_2$ must also be an edge, in order to form a cycle
    (see Figure~\ref{fig:out_neighborhood_forms_a_cycle}).
\begin{figure}
    \centering
    \includegraphics[scale=.95]{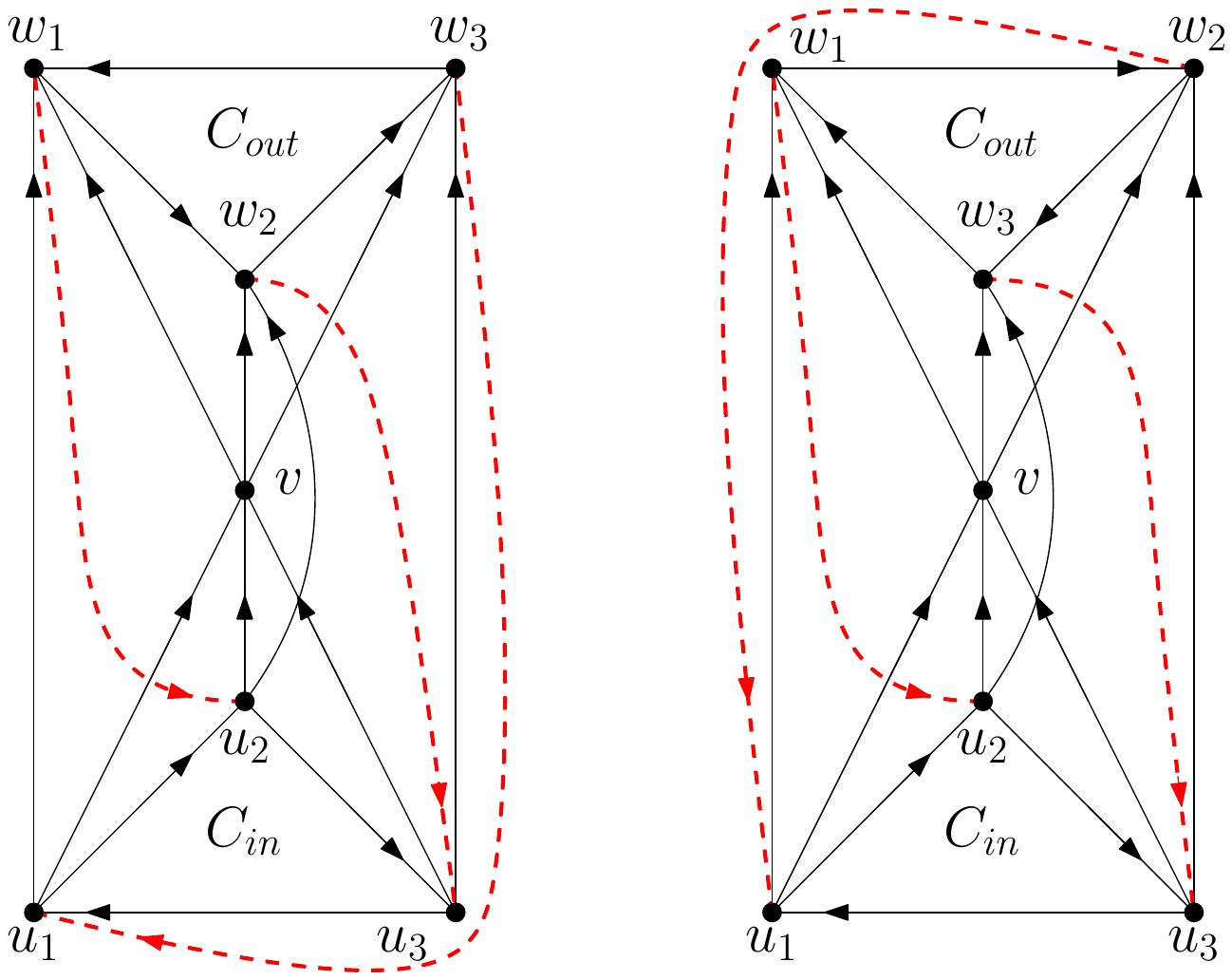}
    \caption{\small The two possible neighborhoods of the vertex $v$ from the proof of  Lemma~\ref{lem:two_graphs}.
        The dashed edges connect $C_{in}$ and $C_{out}$ so that the out-neighbors of each vertex
        from $C_{in}$ form a cycle.
        }
        \label{fig:out_neighborhood_forms_a_cycle}
\end{figure}

By Lemma~\ref{lem:thomassen_inneighborhood_cycle}
    the in-neighborhoods of each vertex from $C_{out}$ forms a cycle as well.
    If the orientations of $C_{in}$ and $C_{out}$ are not the same (i.e.\ $\pi(2)=3$),
    then there is a unique way of adding edges from $C_{in}$ to $C_{out}$
    so that the in-neighborhoods of each vertex from $C_{out}$ form a cycle
    (see the right picture of Figure~\ref{fig:only_two_graphs}).
    The resulting digraph satisfies all the conditions from the statement,
    and in particular all vertices in it have out-degree $3$.
    Thus, when $\pi(2)=3$, this is the unique digraph satisfying the conditions in the lemma.

In the remaining case, when the orientations of $C_{in}$ and $C_{out}$ are the same (i.e.\ $\pi(2)=2$),
    then the in-neighbors of each vertex from $C_{out}$ form a cycle as well
    (see Figure~\ref{fig:out_neighborhood_forms_a_cycle}).
    However, the partial digraph considered thus far still does not satisfy the required conditions,
    as the out-degrees of the vertices in $C_{out}$ and in-degrees of of the vertices in $C_{in}$ are only~$2$.

Consider the vertex $w_1$ in $C_{out}$; it already has two 
    out-neighbors~{$w_2$ and $u_2$}. We claim that its third out-neighbor, denoted by $x$,
    must be a new vertex $x\notin\{v,u_1,u_2,u_3,w_1,w_2,w_3\}$:
    clearly $x\notin\{w_2,u_2\}$ as it is the third out-neighbor of $w_1$;
    also, $x\notin\{v,w_3,u_1\}$ since $G$ has no $2$-cycles;
    {lastly, $x\neq u_3$ since $u_3$ is the out-neighbor of both $u_2,w_2$
    and hence cannot form a cycle with them,
    which would contradict Corollary~\ref{cor:out_neighbors_cycle}
    with respect to the out-neighbors of $w_1$.}
    
Thus, $x$ must be a new vertex. 
    Now, we use the fact that each two consecutive vertices have common out-neighbor
    (Lemma~\ref{lem:common_vertex}) and thus,
    the vertex $x$ has to be an out-neighbor of $w_2$ and $w_3$ as well.
    This fixes the out-degree of the vertices of $C_{out}$.
    We claim that the out-neighbors of $x$ must be $u_1,u_2,u_3$:
    by Lemma~\ref{lem:common_vertex} the vertices $x$ and $w_1$
    have a common out-neighbor, and therefore either $w_2$ or $u_2$
    must be an out-neighbor of $x$.
    However, $w_2$ is excluded as there is no 2-cycle in $G$.
    Thus, $u_2$ is an out-neighbor of $x$
    Repeating the same argument for the pairs $x, w_2$ and $x, w_3$,
    implies that the new vertex $x$ is an in-neighbor of $u_1,u_2,u_3$.
    The obtained digraph satisfies all the conditions in the lemma
    and is therefore the unique solution in the case when $\pi(2)=3$
    (see the left picture of Figure~\ref{fig:only_two_graphs}).
\end{proof}

\begin{figure}
    \centering
    \includegraphics[scale=1]{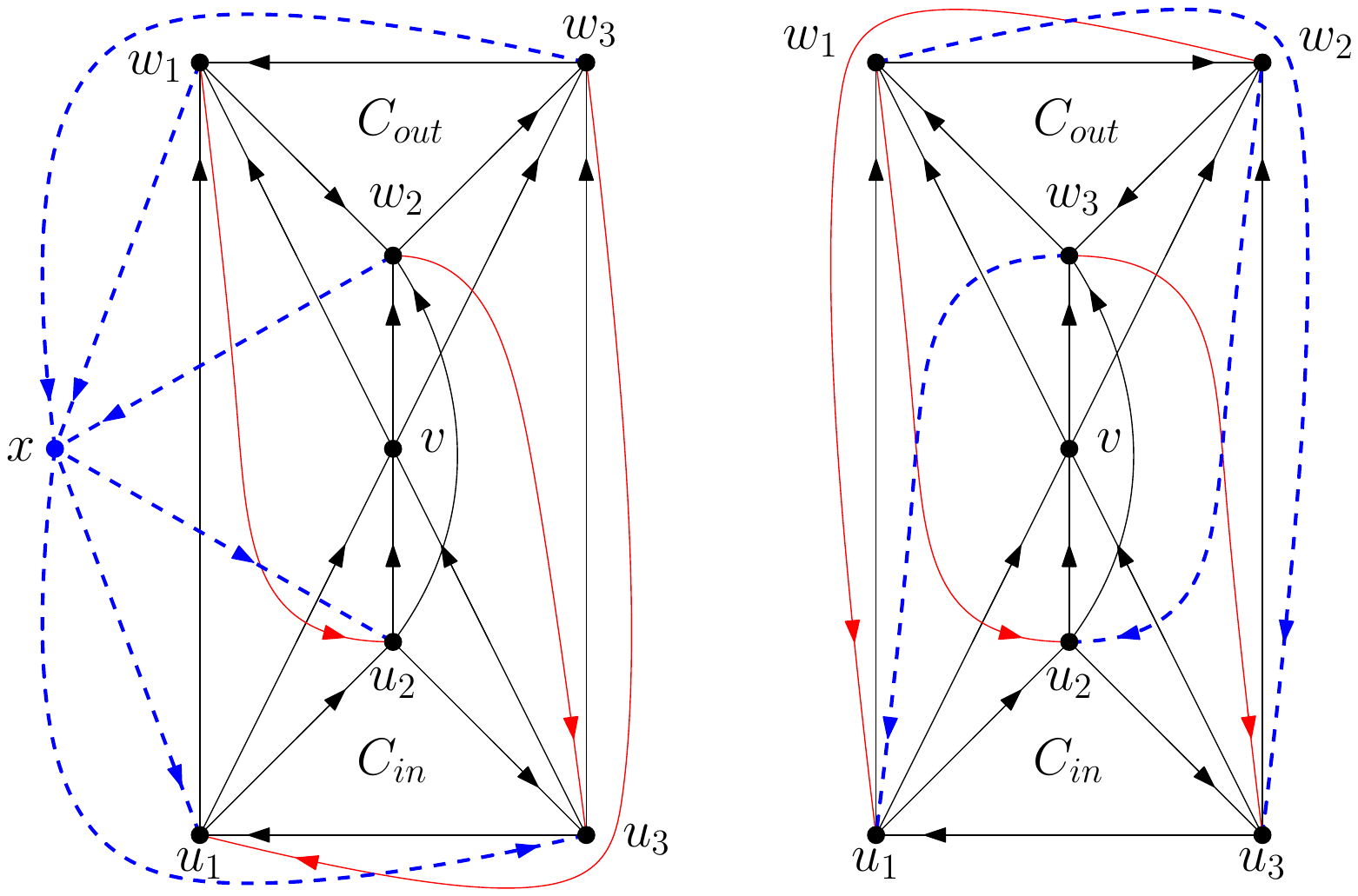}
    \caption[]{\small The two digraphs from Lemma~\ref{lem:two_graphs}.
        Note that while the right digraph basically consists only
        of the neighborhood of $v$,
        the left digraph contains one more vertex.
        Note that both these digraphs have more than one friendly partition:
        e.g.\ in the left digraph $\{x, u_1,u_2,u_3\}$ and its complement as well as
        $\{x, u_1,w_1,w_3\}$ and its complement,
        and in the right digraph $\{u_1,u_2,u_3\}$, 
        and $\{v,w_2,u_2\}$ and their complements. }
    \label{fig:only_two_graphs}
\end{figure}

Lemma~\ref{lem:two_graphs} will be used to show that digraphs with out-degrees $\geq3$
    that do not contain 2-cycles have at least $2$ friendly partitions.
    How about digraphs that do contain 2-cycles?
    If a digraph with out-degrees $\geq 3$ has at least two 2-cycles
    then it has more than one friendly partition as well.
    Indeed, a 2-cycle itself as one part
    and the rest of the digraph as the other part form a friendly partition.
    The next lemma is the key to handle digraphs
    which contain precisely one 2-cycle.
    
\begin{lemma}\label{lem:more_partitions_dom}
    Let $G$ be a digraph with all edges dominated,
    with all vertices of out-degree 3,
    and with \emph{exactly} one 2-cycle.
    Then, $G$ has at least two friendly partitions.
\end{lemma}
\begin{proof}
    We prove by contradiction: assume $G=(V,E)$ is a counterexample
    with the smallest number of vertices,
    and let $C_1 = \{u,v\}$ be the unique 2-cycle in $G$.
    Thus, $C_1$ is one part of the unique friendly partition of $G$.
    
The second part $V_G\setminus C_1$ must contain a cycle $C_2$ 
    (because the minimal out-degree in the digraph induced by it is $\geq 1$).
    Pick $C_2$ to be a cycle with fewest vertices in the second part,
    and let $A=V\setminus(C_1\cup C_2)$ denote the set of remaining vertices.
    Note that the subdigraph induced by~$A$ is acyclic
    (or else it would contain a cycle $C_3$ which is disjoint
    from $C_2$ and Lemma~\ref{lem:extending_friendly_sets}
    would imply the existence of another
    friendly partition, namely the one separating $C_2$ and $C_3$).
    Also note that there are no edges
    from $A$ to $C_1$ -- or else we could move such a vertex 
    to the part of $C_1$ and obtain another friendly partition.
    
Let $C_A$ be the set of vertices in $C_2$
    that have an out-neighbor in $A$,
    and let $C_B$ be the set of vertices in $C_2$
    that have an out-neighbor in $C_1$.
    Note that: (i) $C_A \cup C_B = C_2$ (by minimality of $C_2$), and
    (ii) $C_B\neq\emptyset$ (or else the subdigraph induced by $A\cup C_2$
        would be itself of minimum out-degree 3, and therefore,
        $G$ would have more than one friendly partition).
        We consider two cases:
    \begin{itemize}
        \item Case 1: The set $C_A$ is \underline{nonempty}.
        Let $c\in C_A$ be a vertex whose outneighbor $b$ on $C_2$ is in~$C_B$.
        Let $a$ be an out-neighbor of $c$ from $A$.
        (See Figure~\ref{fig:C_A_nonempty}.)
    
    We claim that there are paths from $a$ to at least 3 distinct vertices in $C_2$.
        Indeed, this follows by taking a vertex $a'\in A$ which is reachable from $a$
        and that has no out-neighbor in $A$ (such $a'$ exists since $A$ is acyclic);
        now, since there are no edges from $A$ to $C_1$, it follows that the 3 out-neighbors
        of $a'$ must be in $C_2$, and are all reachable from $a$.
        Hence, at least one of these vertices is neither $b$ nor $c$.
        Denote such a vertex by $w$.
    
    Consider the cycle $C_3$ which starts at $c$, and continues via $a$ and $a'$
        to $w$, and then continues together with $C_2$ until it reaches back $c$
        (see Figure~\ref{fig:C_A_nonempty}).
        Notice that $b\notin C_3$.
    
    Now, by Lemma~\ref{lem:extending_friendly_sets} one can obtain an additional friendly partition
    whose one part contains $C_1 \cup \{b\}$ and the other contains $C_3$,
    which is a contradiction.

        \begin{figure}
            \centering
            \includegraphics{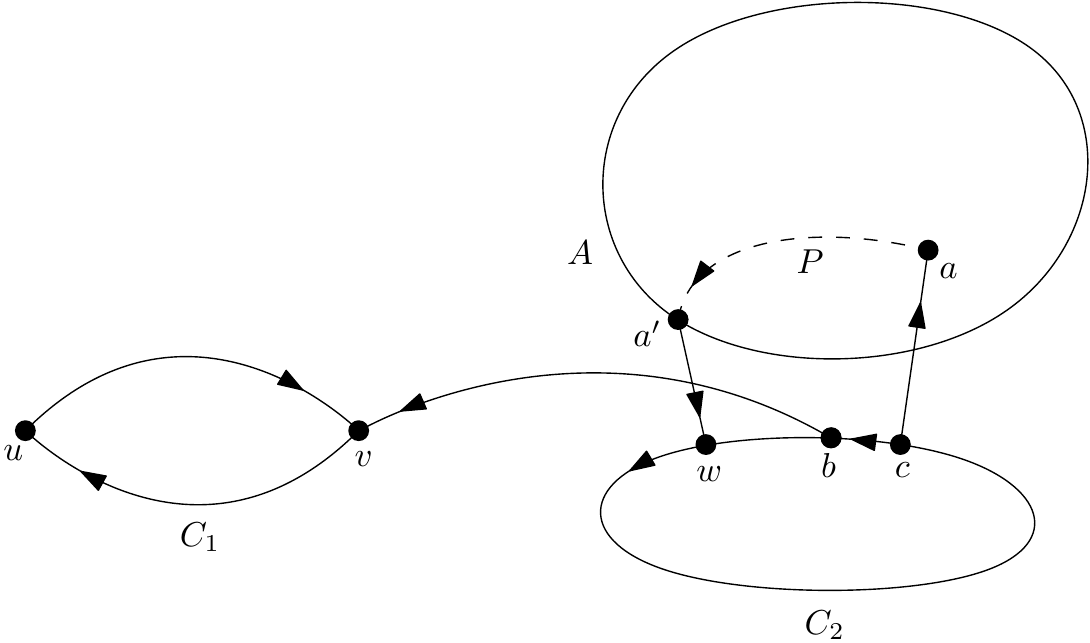}
            \caption{\small A picture illustrating the case when the set $C_A$ from the proof
                of Lemma~\ref{lem:more_partitions_dom} is non-empty.}
                \label{fig:C_A_nonempty}
        \end{figure}
        
        \item Case 2: The set $C_A$ is \underline{empty}.
        In other words, each vertex in $C_2$ has out-going edges to 
        the two vertices $u,v$ in $C_1$.
        We consider two subcases.
        \begin{itemize}
            \item There is an edge from some vertex in $C_1$, say $v$,
                to some vertex $w \in C_2$.
                Thus, $\{v, w\}$ is a 2-cycle distinct from $C_1$.
                However, by the assumption $G$ has exactly one 2-cycle.
                A contradiction.
            \item There is no such edge.
                {Therefore, each of $u,v$ has its two out-neighbors in $A$.
                Let $a$ be a vertex from $A$.
                Since $a$ does not lie in a $2$-cycle,
                Lemma~\ref{lem:thomassen_inneighborhood_cycle} implies
                there is a cycle $K$ formed by the in-neighbors of $a$.
                We claim that $K=C_1$:
                there is no vertex from $C_2$ in $K$, 
                since there is no edge from $C_2$ to $A$.
                Since $A$ is acyclic, $K$ cannot contain vertices only from $A$.
                Since there are no edges from $A$ to $C_1$,
                the cycle $K$ cannot contain vertices from both $A$ and $C_1$.
                
                The only candidate for $K$ is thus $C_1$
                which implies that $A$ contains only two vertices $a_1$ and $a_2$.
                
            Let $w_1$ be an out-neighbor of $a_1$ from $C_2$
                and $w_2$ be an out-neighbor of $a_2$ from $C_2$ different from $w_1$.
                Such $w_1$ and $w_2$ exist since there is no edge from $A$ to $C_1$.
                Recall that $u,v \in  C_1$ are both out-neighbors of $w_1$ and $w_2$.
                Whence, $\{a_1,w_1, v\}$ and $\{a_2, w_2, u\}$ are two disjoint cycles
                different from $C_1$ witnessing
                (by Lemma~\ref{lem:extending_friendly_sets})
                there is another friendly partition;
                a contradiction. (See Figure~\ref{fig:C_A_empty}.)
                }
        \end{itemize}
    \end{itemize}
\end{proof}   

\begin{figure}
    \centering
    \includegraphics{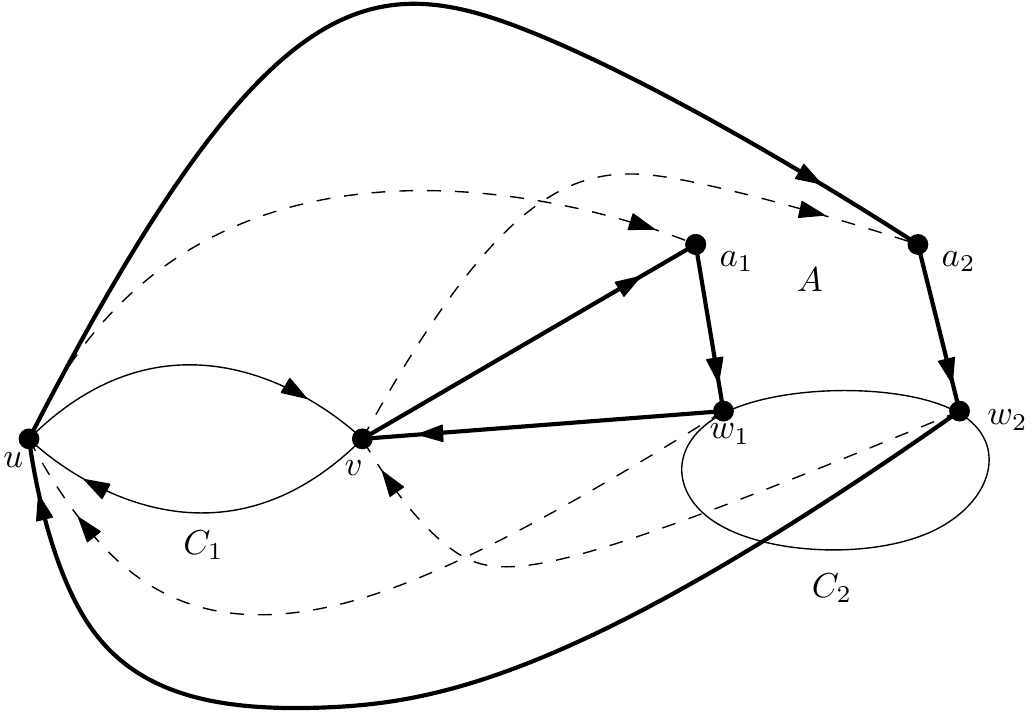}
    \caption{\small A picture illustrating the case of the proof of
        Lemma~\ref{lem:more_partitions_dom}
        when the set $C_A$ is empty and there is no edge from $C_1$ to $C_2$.}
        \label{fig:C_A_empty}
\end{figure}

\begin{proof}[Proof of Theorem~\ref{thm:more_partitions}]
    First of all, we delete edges
    so that all vertices have out-degree exactly 3.
    Then, we contract edges so that 
    in the resulting digraph $G$,
    all edges are dominated or a part of 2-cycle. 
    This procedure preserves the out-degree of 3.
    
If $G$ contains no 2-cycle, then it is isomorphic to one of the two digraphs from
    Lemma~\ref{lem:two_graphs}. Both these digraphs have more than one friendly partition.
    
If $G$ contains exactly one 2-cycle
    then it has more than one
    friendly partition as well by Lemma~\ref{lem:more_partitions_dom}.
    
Finally, if $G$ contains at least two 2-cycles, say $C_1, C_2$,
    then $(C_1,V_G\setminus{C_1}), (C_2,V_G\setminus{C_2})$
    are two distinct friendly partitions by Lemma~\ref{lem:extending_friendly_sets}.
    
As discussed in the beginning of Section~\ref{sec:existence_of_a_separable_vertex}
    these friendly partitions of $G$ induce (distinct) friendly partitions in the original digraph.
\end{proof}

\section{Strongly Connected Digraphs}\label{sec:strongly_connected_graphs}
In this section we prove Proposition~\ref{pro:strongly_connected}
    which says that in order to show
    that every pair of vertices can separated by a friendly partition
    in any digraph $G$ with minimum out-degree $d$,
    it suffices to only consider strongly connected digraphs $G$:

\begin{repproposition}{pro:strongly_connected}[Restatement]
Let $d \geq 3$. If every strongly connected digraph
    with minimum out-degree at least $d$ satisfies
    that each pair of vertices in it is separable,
    then every digraph with minimum out-degree at least $d$ satisfies this property.
\end{repproposition}
\begin{proof}
    Suppose we can separate each pair of vertices
    in strongly connected digraphs with out-degrees $\geq d$.
    Consider a digraph $G=(V,E)$ and its strongly connected components.
    Let us contract each strongly connected component to one vertex
    and denote the resulting digraph~$A$.
    Note that $A$ is acyclic.
    Let $S$ be the set of vertices of $G$
    that are contracted to vertices of out-degree 0 in $A$,
    let $G_S$ denote the subdigraph of $G$ induced by $S$,
    and let $u, v$ be a pair of vertices from $G$. 
    We need to show that $u,v$ are separable; we consider three cases:
    \begin{itemize}
        \item Case 1: $u, v \in S$.
            By the assumption, they are separable in $G_S$
            since it is a disjoint union of strongly connected digraphs
            with minimum out-degree at least $d$.
            This partition can be extended to a friendly partition of the entire digraph $G$ {by Lemma~\ref{lem:extending_friendly_sets}}.
            Thus, in this case $u,v$ are separable in $G$.
        \item Case 2: $u \in V \setminus S$, $v \in S$.
            We distinguish two subcases:
        \begin{itemize}
            \item There is a path $P$ from $u$ to $S$
                such that the $P \cap S = \{w\} \neq \{v\}$.
                Again, by the assumption $v$ and $w$ are separable in $G_S$.
                Now, it is sufficient to extend the partition of $G_S$
                to a partition of $G$ placing $P$ to the part of $w$.
            \item There is no such path.
                Let $V_u$ denote the set of vertices from $V \setminus S$
                which are reachable from~$u$.
                Then $v$ is the only out-neighbor of $V_u$ in $S$
                and thus,~$V_u$ induces a subdigraph~$G_u$
                which is of minimum out-degree at least $d-1 \geq 2$.
                Also $S$ induces a subdigraph $G_S$
                in which all out-degrees are $d\geq 2$.
                Thus, we can extend the disjoint pair $\{V_u, S\}$ to a friendly partition
                of the entire digraph $G$
                {by Lemma~\ref{lem:extending_friendly_sets}.
                Such partition separates $u$ and $v$.}
        \end{itemize}
        
        \item Case 3: $u,v \in V \setminus S$.
            We distinguish two subcases:
        \begin{itemize}
            \item There are {vertex}-disjoint paths $P_u$ from $u$ to some $s_u \in S$
            and $P_v$ from $v$ to some $s_v \in S$.
            {We can assume that $P_u \cap S = \{s_u\}$
            and $P_v \cap S = \{s_v\}$. (Otherwise, replace $P_u,P_v$ by prefixes $P'_u,P'_v$ which satisfy this.)}
            Since $s_v$ and $s_u$ are separable in $G_S$
            we can extend such partition to $G$
            so that $P_u$ is in the part of~$s_u$ and $P_v$ in the part of~$s_v$.
            \item There are no such paths.
            Then, by Menger's Theorem \cite{Menger27}, \cite{Aharoni09}
            \footnote{Menger's theorem for (possibly infinite) digraphs:
            let $A$ and $B$ be two sets of vertices in a possibly infinite digraph.
            Then there exist a family $P$ of disjoint $A \to B$ paths,
            and a set $S$ of vertices separating $A$ from $B$,
            such that $S$ consists of a choice of precisely
            one vertex from each path in P.}
            there exists $t \in V\setminus S$ which separates the pair $\{u,v\}$ from $S$.
            Now, delete all outgoing edges from $t$ and
            add edges $tu, tv$ and edges from $t$
            to some other $d-2$ out-neighbors of $u,v$
            which are different from $t$.
            (In the special case when $t\in\{u,v\}$, say $t=u$,
            add an edge $tv$ and another $d-1$ edges
            from $t$ to $d-1$ out-neighbors of $v$
            which are different from $u$.)
            
            Let us denote the resulting digraph $G^\prime$. 
            Consider the subdigraph $G^{\prime\prime}$ of $G^\prime$
            which is induced by all vertices which are reachable from $u$ or from $v$. 
            We claim that $G^{\prime\prime}$ is strongly connected
            with minimum out-degree at least $d$:
            indeed, to see that its minimum out-degree is at least $d$,
            observe that we only modified the out-neighborhood of $t$
            and we connected it with $d$ vertices which are reachable from $u$ or $v$.
            To see that $G^{\prime\prime}$ is strongly connected,
            let $x,y \in G^{\prime\prime}$ be a pair of vertices.
            Since $t$ separates $\{u,v\}$ from $S$ in $G$,
            and $x$ is reachable from $u$ or from $v$,
            it follows that there must be a path from $x$ to $t$. 
            (Or else, any path from $u$ or $v$ to $S$ via $x$
            would reach $t$ before it reaches $x$,
            which implies that $x\notin V(G^{\prime\prime})$,
            which is a contradiction to the definition of $x$.)
            Since there is an edge from $t$ to both~$u$ and $v$
            and since $y$ is also reachable from $u$ or $v$
            there is a walk from $x$ to $y$ through~$t$. 
        
        Therefore, $G^{\prime\prime}$
            is strongly connected of minimum out-degree at least $d$
            and by the assumption there is a friendly partition of $G^{\prime\prime}$
            separating $u,v$.
            The vertices $u,v$ can also be separated in $G$:
            indeed, we use the partition of $G^{\prime\prime}$
            and add $S$ along with a path from $t$ to $S$ to the same part as $t$.
            This gives us a partial friendly partition {separating $u,v$}
            that can be extended to the entire digraph $G$
            {by Lemma~\ref{lem:extending_friendly_sets}.}
        \end{itemize}
    \end{itemize}
    In all three cases, vertices $u,v$ are separable in $G$.
    This finishes our proof.
\end{proof}

\section{Vertex-Transitive Digraphs}\label{sec:vertex-transitive_graphs}
In this section, we prove our partial results for vertex-transitive digraphs.
    First of all, recall the equivalence relation
    ``being unseparable'' denoted by $\sim$ on the vertex set of a digraph:
    \begin{center}
        ``$u\sim v\quad \iff\quad $   $u$ cannot be separated from $v$
        by a friendly partition.''
    \end{center}
    Observe that that for a vertex-transitive digraph $G=(V,E)$
    each class of $\sim$ has the same number of vertices:
    indeed, let $u,v \in V$ be an arbitrary pair of vertices of $G$
    and let $\tau$ be an automorphism of $G$ mapping $u$ to $v$.
    If $u$ is separable from some vertex $s_u$
    by a friendly partition $\{V_1, V_2\}$
    then $v$ is separable from the vertex $\tau(s_u)$
    by $\{\tau(V_1), \tau(V_2)\}$
    (since $\{V_1, V_2\}$ is a friendly partition
    if and only if $\{\tau(V_1), \tau(V_2)\}$
    is a friendly partition).
    Conversely, if $u$ is separable from some vertex $s_v$
    by a friendly partition $\{V_1, V_2\}$
    then $v$ is separable from the vertex $\tau^{-1}(s_v)$
    by $\{\tau^{-1}(V_1), \tau^{-1}(V_2)\}$
    since $\{V_1, V_2\}$ is a friendly partition
    if and only if $\{\tau^{-1}(V_1), \tau^{-1}(V_2)\}$
    is a friendly partition.
    Therefore, each two vertices have the same number of
    vertices from which they can be separated by a friendly partition.
    This observation implies Proposition~\ref{pro:vertex_transitive_graph_prime}:
\begin{repproposition}{pro:vertex_transitive_graph_prime}[Restatement]
    If the number of vertices in a vertex-transitive digraph
    with degree at least $3$ is prime
    then each pair of vertices is separable.
\end{repproposition}
    Indeed, consider a transitive digraph $G$ with out-degrees $\geq 3$ and with a prime number $n$ of vertices.
    Then, since $G$ has a non-trivial friendly partition
    it follows that all equivalence classes with respect to $\sim$
    have size $< n$. 
    Thus, since all equivalence classes have the same size
    it follows that each equivalence class is a singleton,
    and consequently that each pair of vertices is separable.

We continue with several lemmas which lead to the proof of
    Proposition~\ref{pro:independent_sets}.
\begin{repproposition}{pro:independent_sets}[Restatement]
    Let $\sim$ denote the following equivalence relation
    on the set of vertices of a digraph $G$.
    \begin{center}
        ``$u\sim v\quad \iff\quad $   $u$ cannot be separated from $v$
        by a friendly partition.''
    \end{center}
    Then, if $G$ is vertex-transitive with degree $d\geq 3$
    then each equivalence class of $\sim$ is an independent set. 
    Moreover, each vertex has its out-neighbors in at least $3$ different classes.
\end{repproposition}
\begin{lemma}\label{lem:special_partition_imlies_singleton}
    Let $G = (V,E)$ be a digraph with all vertices of out-degree $\geq 3$.
    If there is a {friendly} partition $\{V_1, V_2\}$
    such that there exists $v_0 \in V_1$ with at least two out-neighbors in $V_1$
    and at least one out-neighbor in $V_2$,
    then there is an equivalence class of size 1.
\end{lemma}
\begin{proof}
    Let us call an in-neighbor $u$ of a vertex $v \in V_1$ \emph{critical in-neighbor}
    if $u \in V_1$ and $v$ is the only out-neighbor of $u$ in $V_1$.
    If the vertex $v$ has at least one out-neighbor in $V_2$ and no critical in-neighbors,
    {then its equivalence class under $\sim$ is $\{v\}$ (in other words, it can be separated from any other vertex):
    indeed, this implies that the partition $\{V_1 \setminus \{v\}, V_2 \cup \{v\}\}$ is friendly,
    and so every vertex in $V\setminus\{v\}$ is separated from $v$ by $\{V_1,V_2\}$ or by $\{V_1 \setminus \{v\}, V_2 \cup \{v\}\}$.}
        
Suppose for contradiction that there is no equivalence class of size 1.
    Since $v_0$ has at least one out-neighbor in $V_2$
    and it is not a singleton,
    there must exist a critical in-neighbor of $v_0$ by the observation above.
    Let us denote it $v_1$.
    Analogously there must be a critical in-neighbor of $v_1$ etc.
    Therefore, there is a sequence $\sigma = (v_0, v_1, v_2, ...)$
    of vertices from $V_1$
    such that $v_{i+1}$ is a critical in-neighbor of $v_i$.
    This sequence must be infinite, otherwise the last element would be singleton.
    We consider two cases and we show that both lead to a contradiction:
\begin{itemize}
    \item Case 1: The first repeated vertex of $\sigma$ is $v_0$.
        That means $v_0$ is a critical in-neighbor of $v_i$ for some $i > 0$.
        This is a contradiction
        since $v_0$ has more than one out-neighbor in $V_1$ 
        by the assumption
        and thus, it cannot be a critical in-neighbor.
    \item Case 2: The first repeated vertex of $\sigma$ is $v_i$ for $i > 0$.
        In other words, there exist $0 < i < j$ such that $v_i = v_j$.
        That means $v_i = v_j$ must have at least 2 out-neighbors,
        $v_{i-1}$ and $v_{j-1}$, in $V_1$.
        (The vertex $v_{i-1}$ must be different from $v_{j-1}$,
        or else $v_i$ would not be the first repeated vertex.)
        However, the only vertex from the sequence $\sigma$
        having more than one out-neighbor in $V_1$ is $v_0$;
        a contradiction.
\end{itemize} 
\end{proof}

Note that Lemma~\ref{lem:special_partition_imlies_singleton} holds in general,
not only for vertex-transitive digraphs.
However, it implies the following corollary for vertex-transitive digraphs
since in such digraphs all classes of the equivalence~$\sim$ have the same size.

\begin{corollary}\label{cor:special_partition_imlies_separability}
    Let $G = (V,E)$ be a vertex-transitive digraph of out-degree $\geq$ 3.
    If there is a partition $V = \{V_1, V_2\}$
    such that there is a vertex from $V_1$ with at least two out-neighbors in $V_1$
    and at least one out-neighbor in $V_2$,
    then all pairs of vertices of $G$ are separable.
\end{corollary}

Now, we show that each equivalence class of $\sim$ for vertex-transitive digraphs
with all vertices of out-degree at least 3 is an independent set.
First of all, we show that each class is a cycle
or an independent set (Lemma~\ref{lem:eqv_cycle}).
Then, we show that it cannot be a cycle (Lemma~\ref{lem:eqv_cannot_be_a_cycle}).

\begin{lemma}\label{lem:eqv_cycle}
    If an equivalence class of $\sim$ in a vertex transitive digraph
    whose out-degrees are $d\geq 3$ contains an edge,
    then each equivalence class is a cycle.
\end{lemma}
\begin{proof}
    Let us fix an arbitrary equivalence class $K$
    and let $G_K$ denote the subdigraph induced by $K$.
    For the rest of the classes the result follows from vertex-transitivity
    since $G_K \cong \tau(G_K)$ for each automorphism $\tau$ of $G$.
        
First of all, note that $G_K$ is vertex-transitive, and in particular all vertices in $G_K$ have the same out-degree $g\geq 1$:
    indeed, let $u,v$ be arbitrary pair of vertices of $G_K$
    and let $\tau$ be an automorphism of $G$ mapping~$u$ to~$v$.
    Since $u,v\in K$ are from the same equivalence class, we have
    $K = \tau(K)$, and therefore $\tau\vert_{K}$ is an automorphism of $G_k$ mapping $u$ to $v$.

Observe that if $g \geq 3$
    then $G_K$ itself can be partitioned in a friendly manner
    by Theorem~\ref{thm:thomassen} and Lemma~\ref{lem:extending_friendly_sets}.
    This is impossible since all pairs of vertices of $K$ are not separable.
        
Now, assume $g = 2$. Pick an arbitrary vertex $v$ from $K$.
    Then, the subdigraph induced by $K \setminus \{v\}$
    has all vertices of out-degree at least $1$.
    Therefore, it contains a cycle $C_1$.
    Let~$u$ be an out-neighbor of $v$ which is not in $K$,
    and let $K_u$ denote the equivalence class of $u$.
    By vertex-transitivity, the out-degree of all vertices
    in the subdigraph induced by $K_u$ is $2$ as well
    and thus, it must contain a cycle $C_2$
    and a path $P \subseteq K_u$ connecting $u$ with $C_2$
    (possibly consisting only of $u$, when $u$ is in $C_2$).
    Hence, there are two disjoint cycles $C_1 \subseteq K,
    C_2 \subseteq K_u$ and the path $\{v\} \cup P$ connecting $v$ and $C_2$. 
    Moreover, $C_1$ does not contain $v$,
    which means that $v \in K$ is separable from $C_1 \subseteq K$
    by some partition extending $C_1, C_2 \cup P \cup \{v\}$
    by Lemma~\ref{lem:extending_friendly_sets};
    a contradiction.
        
The only remaining possibility is $g = 1$.
    In this case, since $G_K$ is vertex-transitive, 
    it must be a union of disjoint cycles.
    Notice however that there can only be one cycle in $G_K$: 
    indeed, by Lemma~\ref{lem:extending_friendly_sets},
    disjoint cycles can be separated by a friendly partition.
\end{proof}

\begin{lemma}\label{lem:eqv_cannot_be_a_cycle}
    An equivalence class of $\sim$ in a vertex transitive digraph $G$
    whose out-degrees are  $d\geq 3$ cannot be a cycle.
\end{lemma}
\begin{proof}
    We prove by contradiction: consider a vertex-transitive graph $G$ with out-degrees $d\geq 3$ in which some equivalence class is a cycle.
    Further assume that $G$ has the minimum number of vertices among all such graphs.
    Notice that by vertex-transitivity all equivalence classes in $G$ are cycles. 
    Let $G_\sim$ be the digraph whose vertex set consists of the
    equivalence classes of $G$
    and there is an edge from class $K$ to class $L$ in $G_\sim$
    if and only if there are $v_K \in K$ and $v_L \in L$
    forming an edge $v_K \to v_L$ in $G$.
    We consider two cases:
\begin{itemize}
    \item Case 1: 
    There is a vertex 
        in $G_\sim$,
        corresponding to an equivalence class $K$ from $G$,
        whose out-degree is $\geq 2$.
        Thus, there are $u,v \in K$
        such that $u\neq v$ (recall that $K$ has at least two vertices
        since it is a cycle)
        and classes $K_u \neq K_v$ different from $K$
        such that $u$ has an out-neighbor in class $K_u$
        and $v$ has an out-neighbor in class $K_v$
        (because each vertex in $K$ has $d-1 \geq 2$ out-neighbours outside $K$).
        Now, we can easily separate vertices $u$ and $v$
        by Lemma~\ref{lem:extending_friendly_sets}
        since the classes $K_u$ and $K_v$ are cycles. 
        This contradicts unseparability of the vertices
        $u$ and $v$ from the same class $K$.
    \item Case 2:
    All out-degrees in $G_\sim$ are 1.
    Thus, $G_\sim$ contains a cycle $C$ with no edge going outside from it.
    Consider a subdigraph $G_C$ of $G$ induced by the vertices from the unions of classes in $C$.
    Note that all out-degrees in $G_C$ are $d$
    since there is no outgoing edge from $C$ in $G_\sim$.
    We claim that also all in-degrees in $G_C$ are $d$.
    Indeed, $G$ is transitive and hence all in-degrees in it are $d$
    Therefore, since all out-degrees in $G_C$ are $d$,
    and no vertex in $G_C$ can have in-degree greater than $d$,
    it follows that all in-degrees in $G_C$ must be $d$.
    
    We first claim that the cycle $C$ must contain more than $2$ vertices:
    indeed, otherwise we reach a contradiction since
    Theorem~\ref{thm:more_partitions} implies that $G_C$ has at least two friendly partitions:
        one partition being $\{C_1, C_2\}$ but every other partition has to separate vertices from $C_1$.
        By Lemma~\ref{lem:extending_friendly_sets} this partition extends to a friendly partition of $G$,
        which separates vertices from the same equivalence class which is a contradiction.

    Thus, we can assume that the cycle $C$ from $G_\sim$ contains at least 3 vertices.
    Pick classes $K,L$ such that $K \to L$ is an edge in $C$.
    By Hall's marriage theorem \cite{Hall35}
        \footnote{A bipartite graph $G$ with parts $X, Y$ of the same size 
        has a perfect matching in $G$ if and only if for every $W\subseteq X$,
        $\lvert W\rvert \leq \lvert N_G(W)\rvert$.
        Moreover, each $d$-regular bipartite graph satisfies this condition:    
        indeed, let $E_W$ be a set of edges between $W$ and $N_G(W)$. Then,
        $d\lvert W\rvert = \lvert E_W\rvert  \leq d\lvert N_G(W)\rvert\implies \lvert W\rvert \leq \lvert N_G(W)\rvert$.}
        there exists a perfect matching $M$ between $K$ and $L$ which consists of edges from $K$ to $L$ in $G_C$:
        indeed, this follows by applying Hall's Theorem to the graph whose sides are the classes $K,L$
        and whose edges are the edges between $K$ and $L$, without their orientation.
        This graph is $(d-1)$-regular and hence contains a perfect matching, because $d-1>0$.
        
        Now, remove the edges between $K$ and $L$
        which are not in $M$ as well as the edges between vertices in $K$.
        Thus, each vertex in $K$ now has in-degree $d-1$ and out-degree $1$
        while each vertex in $L$ has out-degree $d$ and in-degree $1$.
        Next, contract the edges of $M$ and note that both out-degree and in-degree
        of each vertex is again $d$ in the resulting graph.
        Further, the number of equivalence classes is decreased by $1$
        since the vertices from $K$ are identified with the vertices
        of $L$.%
        \footnote{
        From the viewpoint of $G_\sim$,
        this corresponds to a contraction of one edge of the cycle $C$.}
        (See Figure~\ref{fig:eq_class_cannot_be_a_cycle}.)
        Moreover, every friendly partition of the resulting digraph
        can be extended to a friendly partition of $G$ by putting identified vertices in the same part.
        Thus, the resulting graph has fewer vertices than $G$,
        and still satisfies that at least one of its equivalence classes
        is a cycle. This contradicts the minimality of $G$.
\end{itemize}
\end{proof}
    
\begin{figure}
    \centering
    \includegraphics[scale=.971]{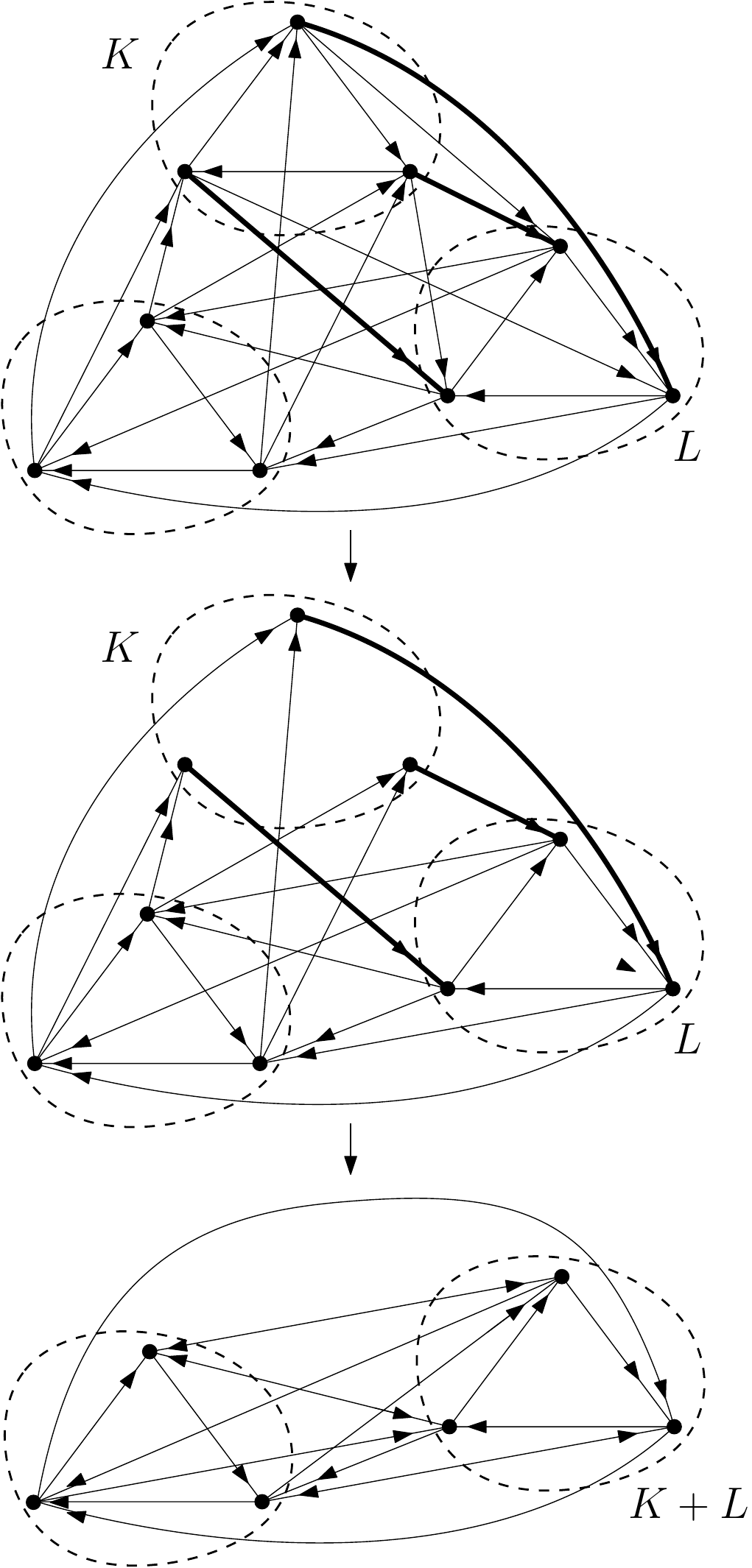}
    \caption{\small A picture showing how we reduce the number of equivalence classes
        using a perfect matching given by Hall's marriage theorem
        in the proof of Lemma~\ref{lem:eqv_cannot_be_a_cycle}.}
        \label{fig:eq_class_cannot_be_a_cycle}
\end{figure}

We are now ready to wrap-up the proof of Proposition~\ref{pro:independent_sets}:
\begin{proof}[Proof of Proposition~\ref{pro:independent_sets}]
    By Lemma~\ref{lem:eqv_cycle}, if there is an equivalence class contains an edge
    then it is a cycle, and by Lemma~\ref{lem:eqv_cannot_be_a_cycle} an equivalence class cannot be a cycle.
    Thus, each equivalence class must be an independent set.
    
    It remains to prove that each vertex has its out-neighbors in at least 3 different classes:
\begin{itemize}
    \item Case 1: If some vertex $v$ has all its out-neighbors
        in the same equivalence class $K$
        then $v$ cannot be separable from its out-neighbors.
        This implies that $v \in K$ which is impossible
        since every equivalence class is an independent set.
    \item Case 2: If $v$ has its out-neighbors in exactly two classes $K, L$
        then in one of them, say in~$K$, there are at least two out-neighbors $x,y$ of $v$.
        Consider a friendly partition separating $v$ from some vertex $u \in L$.
        (Note that $v\notin L$ because $L$ is independent and $v$ has an out-neighbor in $L$.)
        Let $V_v$ be the part containing $v$ and $V_u$ the part containing $u$.
        All vertices of $L$ must be in $V_u$ since $u$ is unseparable from them
        and all vertices from $K$ must be in $V_v$ (or else $v$ would have
        no out-neigbor in $V_v$).
        However, by Corollary~\ref{cor:special_partition_imlies_separability} 
        all pairs of vertices in $G$ are separable since $v$ has at least two out-neigbors, $x,y$, in its part $V_v$ and at least one neigbor, $u$,
        in the other part $V_u$.
        In other words, all classes are singletons which contradicts the fact that $x,y$ are in the same class.
    \end{itemize}
    Therefore, each vertex has its out-neighbors in at least 3 different classes. 
\end{proof}

\section{Infinite Digraphs}\label{sec:infinite_graphs}
This last section consists of the proof of Proposition~\ref{pro:infinite_graphs}.
    
\begin{repproposition}{pro:infinite_graphs}[Restatement]
    Let $G$ be a (possibly infinite) digraph with minimum out-degree at least $3$.
    Then, there exists a friendly partition in $G$.
\end{repproposition}

    We start with the following variant of König's lemma~\cite{Konig27}.
    
\begin{lemma}\label{lem:infinite_path}
    Let $G(V,E)$ be an infinite digraph with bounded out-degree and with no cycle.
    If there is a vertex $v \in V$ such that there are infinitely many
    vertices reachable from it, then $G$ contains an infinite path.
\end{lemma}
\begin{proof}
    We can construct such infinite path inductively as follows.
    We start with $P^{(1)} = \{v\}$
    and in each step we extend the path $P^{(i)}$ to $P^{(i+1)}$
    by one vertex assuming that the set of reachable vertices
    from the last vertex of $P^{(i)}$ is infinite
    (in the  beginning such assumption is guaranteed
    by the assumption from the statement).
    We can always choose a next vertex preserving our assumption
    (by the infinite pigeonhole principle)
    since out-degree of each vertex is bounded.
    Note that no vertex is repeated since $G$ contains no cycle.
\end{proof}
\begin{proof}[Proof of Proposition~\ref{pro:infinite_graphs}]
    It will be convenient to assume that all out-degrees are exactly 3.
    This is without loss of generality since by the Axiom of choice we can choose exactly three out-neighbors
    of each vertex and remove all of its other outgoing edges.
    
Let $V$ denote the vertex set of $G$.
    Observe that if there are disjoint sets $P_1, P_2 \subset V$
    such that they induce subdigraphs of out-degree at least 1
    then we can extend them to a friendly partition in the same way
    as we can do it for finite digraphs (Lemma~\ref{lem:extending_friendly_sets}).

Our proof strategy is to show there exist disjoint sets $P_1, P_2 \subset V$
    such that for each $i\in\{1,2\}$ either there is a directed cycle containing all vertices of $P_i$
    or there is an infinite path containing all vertices of $P_i$,
    each of which induces a cycle or an infinite path.
    Both a cycle and an infinite path have out-degrees 1 and hence
    this guarantees the existence of a friendly partition of $G$.

For a vertex $u$ in $G$, let $R_u$ denote the set of vertices reachable from a vertex $u$.
    If there is a vertex $u$ such that $R_u$ is finite
    then $R_u$ induces a finite subdigraph whose out-degrees are all~3.
    By Theorem~\ref{thm:thomassen} this finite digraph contains two disjoint cycles 
    and we are done.

Otherwise, pick an arbitrary pair of distinct vertices $u$ and $v$.
    Let $V_k \subseteq R_u \cup R_w$ denote the set of vertices
    with distance \emph{at-least} $k \in \mathbb{N}$ from \emph{both} $u$ and~$v$.
    (So, $V_1= R_u\cup R_w$ is the set of vertices reachable from $u$ \emph{or} $v$.)
    Note, that for all $k$, $V_k\neq\emptyset$, because both $R_u$ and $R_v$ are infinite
    and because all the out-degrees are $3$.
    We consider two cases:
\begin{itemize}
    \item Case 1: There exists $k$ for which there do not exist two disjoint paths,
        one starting in $u$ and one in $v$ that end in $V_k$.
        Therefore, by Menger's theorem for infinite graphs
        \cite{Menger27}, \cite{Aharoni09}
        there exists a vertex $t$ separating $\{u,v\}$ from $V_k$.
        Let us delete the vertex $t$ from $G$
        and let $V_{u,v}$ denote the set of vertices reachable from $u$ or $v$
        in the resulting digraph $G \setminus \{t\}$.
        Note that $V_{u,v}$ is finite and that it induces a subdigraph with minimum out-degree at least 2.
        Such subdigraph contains a cycle $C_1$.
        
    The set $V \setminus V_{u,v}$ induces an infinite subdigraph
        in which there are infinitely many vertices reachable from $t$:
        indeed, in $G$ there are infinitely many vertices reachable from $u$ \new{or} from $v$,
        and $t$ separates the finite set $V_{u,v}$ containing $u$ and $v$ from the infinite set $V \setminus \left( V_{u,v} \cup \{t\}\right)$
        containing the remaining vertices reachable from $u$ \new{or} from $v$.
        Therefore, Lemma~\ref{lem:infinite_path} implies that the subdigraph induced by $V \setminus V_{u,v}$
        contains a cycle $C_2$ or an infinite path $P$.
        Consequently, $C_1$ and either $C_2$ or $P$
        witness the existence of a friendly partition in the original graph $G$.
    \item Case 2: For every $k$, there are two disjoint paths from $u$ and $v$ to $V_k$.
        We prove that this implies the existence of two infinite disjoint paths
        starting in $u$ and $v$, respectively.
        We construct such paths in the following way
        (which is similar to the proof of Lemma~\ref{lem:infinite_path}):
        we start with $P^{(1)}_u = \{u\}, P^{(1)}_v = \{v\}$ and in each step,
        we extend the paths $P^{(i)}_u, P^{(i)}_v$
        to $P^{(i+1)}_u, P^{(i+1)}_v$, each by one vertex.
    
    Let $\mathcal{P}^{(i)}$ be the set of all pairs of finite disjoint paths
        of the same length with prefixes $P^{(i)}_u$ and $P^{(i)}_v$,
        respectively.
        In the base-case, $\mathcal{P}^{(1)}$ is infinite by the assumption in Case 2.
        We construct the two disjoint paths by maintaining the invariant
        that $\mathcal{P}^{(i)}$ is infinite for all $i$.
        Indeed, assume we have done so for $i$ steps:
        so, $\mathcal{P}^{(i)}$ is infinite, which in particularly implies that the constructed paths $P^{(i)}_u$ and $P^{(i)}_v$ are disjoint.
        Now, there are only finitely (in fact, $9$) possible ways
        to extend $P^{(i)}_u$ and $P^{(i)}_v$ since the out-degree of each vertex is $3$.
        Therefore, since $\mathcal{P}^{(i)}$ is infinite, 
        the infinite pigeon principle implies that there is at least one extension for which $\mathcal{P}^{(i+1)}$
        is infinite.
        In this way, we can extend $P^{(i)}_u$ and $P^{(i)}_w$ indefinitely.
        This pair of disjoint paths then witnesses the existence of a friendly partition of $G$
        as required.
\end{itemize}
\end{proof}

\section{Appendix}\label{sec:appendix}
Recall that Question~\ref{q:r_friendly_partition} asks whether for every
$r > 1$ there is a $d(r)$ such that each digraph with minimum out-degree at least $d(r)$ has a non-trivial $r-$friendly partition.

Note that an affirmative answer to this question has the following nice corollary:
each digraph on $n$ vertices with minimum out-degree at least $d(r)$ contains a subdigraph
with minimum out-degree $r$ on at most $\frac{n}{2}$ vertices.

While we do not know the answer to Question~\ref{q:r_friendly_partition},
we present a proof by Ron Holzman of the above corollary.
We focus on the case $r=2$ but the idea applies more generally.

\begin{proposition}[Ron Holzman, personal communication]\label{pro:ron}
    If $G$ is a digraph on $n$ vertices with all out-degrees 10
    then it contains a subdigraph on at most $\frac{n}{2}$ vertices
    with all out-degrees at least 2.
\end{proposition}
\begin{proof}
    We select every vertex independently with probability $\frac{1}{3}$
    and denote by $X$ the random set of these selected vertices.
    Let $Z$ be the set of vertices of $G$ having no out-neighbor in $X$,
    and $O$ be the set of vertices of $G$ having one out-neighbor in $X$.
    We create a set of vertices $Y$
    by adding to $X$ two out-neighbors of each vertex in $Z$
    and one out-neighbor of each vertex in $O$.
    Such $Y$ contains at least two out-neighbors of every vertex of $G$.
    We proceed to upper bound the expected size of $Y$.
    \begin{align*}
        \text{E}(|Y|) \leq \text{E}(|X|) + 2 \text{E}(|Z|) + \text{E}(|O|) =
        n \left(\frac{1}{3} + 2\left(\frac{2}{3}\right)^{10} +
        \frac{10}{3}\left(\frac{2}{3}\right)^{9}\right) < 0.46n.
    \end{align*}
    Therefore, there exists a realization of $Y$ of size less than $0.46n < \frac{n}{2}$.
    Since $Y$ contains at least two out-neighbors of every vertex of $G$
    it has to itself induce a subdigraph with minimum out-degree at least 2.
\end{proof}

One can extend this result for $r > 2$
    by choosing suitable $d=d(r)$ so that
\begin{align*}
    \frac{1}{3} + \sum_{k=0}^{r} (r-k)\binom{d(r)}{k} \left(\frac{1}{3}\right)^k \left(\frac{2}{3}\right)^{d(r)-k}
    < \frac{1}{2}.
\end{align*}
The right hand side is bounded from above by
$1/3 + (r+1)\cdot r\cdot  d(r)^r \cdot \left(2/3\right)^{d(r)-r}$.
Thus, it is sufficient to find $d(r)$ such that
\begin{align*}
    (r+1)\cdot r\cdot  d(r)^r \cdot \left(\frac{3}{2}\right)^r\left(\frac{2}{3}\right)^{d(r)} &< 1/6\\
\end{align*}
which is possible since the left-hand side converges to zero as $d(r)\to\infty$.

\bibliographystyle{alpha} 

\end{document}